%% file: dissipation_qrc.tex
\PassOptionsToPackage{hyphens}{url}
\documentclass[a4paper,twocolumn,11pt,unpublished]{quantumarticle}
\pdfoutput=1
\pdftrailerid{}
\usepackage[utf8]{inputenc}
\usepackage[english]{babel}
\usepackage[T1]{fontenc}
\usepackage{amsmath,amssymb}
\usepackage{amsthm}
\usepackage{booktabs}
\usepackage{enumitem}
\usepackage[skins,breakable]{tcolorbox}
\usepackage{graphicx}
\usepackage{placeins}
\usepackage{float}
\usepackage{needspace}
\usepackage{microtype}
\usepackage[numbers,sort&compress]{natbib}
\usepackage{hyperref}
\usepackage{orcidlink}

\AtBeginDocument{\normalfont}

\hypersetup{allcolors=quantumviolet}
\makeatletter
\renewcommand*\paragraph{\@startsection{paragraph}{4}{\z@}{1.6ex \@plus .6ex \@minus .2ex}{-1em}%
  {\normalfont\normalsize\bfseries}}
\renewcommand{\@printtitle}{%
  \color{quantumviolet}%
  \@titleatfontsize{\fontsize{17.6}{21}\selectfont}\par\vspace{0.25em}%
}
\def\@printauthor#1#2{%
  \def\footnote{\ClassError{quantumarticle}{You must not put a
  \string\footnote{} command inside the argument of \string\author}{}}%
  \mbox{%
    \ifcsdef{author#1orcid}{%
      \href{https://orcid.org/\csname author#1orcid\endcsname}%
        {\color{black}#2}\,\orcidlink{\csname author#1orcid\endcsname}%
    }{#2}%
  }%
  \ifnumcomp{\the@affiliationcounter}{>}{1}{%
    \textsuperscript{\forlistcsloop{\@@commaspacebefore}%
      {author#1affiliations}\unskip}%
  }{}%
  \ifnumcomp{#1}{<}{\the@authorcounter-1}{, }{%
    \ifnumcomp{#1}{=}{\the@authorcounter-1}{%
      \ifnumcomp{#1}{=}{1}{ and~}{, and~}%
    }{}%
  }%
}
\def\@printauthorextrainfo#1{%
  \vspace{0.55\baselineskip}%
  \raggedright
  \textit{Corresponding author}\par
  \csname @authorname#1\endcsname
  \ifcsdef{author#1orcid}{%
    \,\orcidlink{\csname author#1orcid\endcsname}%
  }{}\enspace
  \ifcsdef{author#1emails}{%
    \forlistcsloop{\@@spaceafter}{author#1emails}%
  }{}%
  \ifcsdef{author#1homepages}{%
    \ifcsdef{author#1emails}{\unskip, }{}%
    \forlistcsloop{\@@spaceafter}{author#1homepages}%
  }{}%
  \ifcsdef{author#1thanks}{%
    \ifcsdef{author#1emails}{%
      \ifcsdef{author#1homepages}{\unskip, }{\unskip, }%
    }{%
      \ifcsdef{author#1homepages}{\unskip, }{}%
    }%
    \forlistcsloop{\@@spaceafter}{author#1thanks}%
  }{}%
  \par\vspace{-0.55\baselineskip}%
}
\makeatother

\newtcolorbox{keybox}{enhanced,breakable,sharp corners,boxrule=0pt,leftrule=1.5pt,
  colback=quantumviolet!4,colframe=quantumviolet!70,left=6pt,right=6pt,top=5pt,bottom=5pt,
  before skip=10pt plus 2pt,after skip=10pt plus 2pt}
\newtheoremstyle{keyresult}{0pt}{0pt}{\normalfont}{}{\bfseries}{.}{0.5em}{}
\theoremstyle{keyresult}
\newtheorem{theorem}{Theorem}

\newcommand{\Tr}{\operatorname{Tr}}

\hypersetup{pdftitle={The form of dissipation shapes what a quantum reservoir remembers},
 pdfsubject={Dissipation engineering, dark states and the memory of quantum reservoirs}}
\begin{document}
\title{The form of dissipation shapes what a quantum reservoir remembers}

\author[qar,mobile]{Markus Baumann}
\orcid{0009-0007-3575-1006}
\email{markus.baumann@campus.lmu.de}
\author[mobile]{Itamar Fink}
\orcid{0009-0000-2262-4445}
\author[mobile]{Johannes Wittmann}
\orcid{0009-0009-9240-7257}
\author[qar,mobile]{Claudia Linnhoff-Popien}
\orcid{0000-0001-6284-9286}
\author[qar,mobile]{Jonas Stein}
\orcid{0000-0001-5727-9151}
\affiliation[qar]{Quantum Applications and Research Laboratory (QAR-Lab),
Ludwig-Maximilians-Universit\"at M\"unchen, Munich, Germany}
\affiliation[mobile]{Mobile and Distributed Systems Group,
Ludwig-Maximilians-Universit\"at M\"unchen, Munich, Germany}

\input{sections/abstract}
\maketitle
\hypersetup{pdftitle={The form of dissipation shapes what a quantum reservoir remembers}}
\raggedbottom
\makeatletter\@secpenalty=0\makeatother

\input{sections/intro}
\input{sections/results}
\input{sections/discussion}

\acknowledgments{The authors thank the Quantum Applications and Research Laboratory and the Mobile and
Distributed Systems Group at Ludwig-Maximilians-Universit\"at M\"unchen for institutional support.}

\Needspace{6\baselineskip}
\section*{Data availability}
The code, the data behind every figure, and the derivation and certificates of the three-spin limits are at
\url{https://github.com/eybmits/what-a-dissipative-quantum-system-remembers}.

\Needspace{6\baselineskip}
\section*{Author contributions}
M.B. contributed to conception, methodology, software, analysis, visualization and writing, I.F. and J.W. to
methodology, interpretation and editing, and C.L.-P. and J.S. to supervision and conceptual development. OpenAI Codex
and Claude Code (Anthropic) assisted with manuscript preparation, checks of derivations and references, and simulation,
verification and plotting code. The authors reviewed the outputs and take responsibility for the final content.

\appendix
\counterwithin*{equation}{section}
\renewcommand{\theequation}{\thesection\arabic{equation}}
\renewcommand{\theHequation}{\thesection.\arabic{equation}}
\input{sections/methods}
\input{sections/appendix}

\begingroup
\footnotesize\microtypesetup{expansion=false}
\bibliographystyle{quantum}
\bibliography{references}
\endgroup

\end{document}

%% file: sections/abstract.tex
\begin{abstract}
Dissipation sets how long a quantum reservoir computer remembers its past inputs, and its rate is the usual control.
Here we show that its form matters as much, for how a memory is stored, written and read. When the qubits decay
independently, stronger loss shortens the memory. When they decay collectively, the quantum Zeno effect holds the
system in dark states that do not decay, and stronger loss prolongs the memory. A longer memory is harder to write,
because it leaves more time for new inputs to mix with the inputs already stored. We prove that the number of inputs
mixed into a stored signal is, on average, the memory lifetime times the fraction of the signal that one input mixes.
We also prove that this number stays fixed, however long the memory lasts, when the input amplitude falls as the
inverse square root of the lifetime. A weaker input leaves a fainter signal to read out, so the form of the
dissipation, the input amplitude and the measurement budget have to be chosen together.
\end{abstract}

%% file: sections/intro.tex
\section{Introduction}
\label{sec:introduction}
\begin{figure*}[t]
\centering
\includegraphics[width=\textwidth]{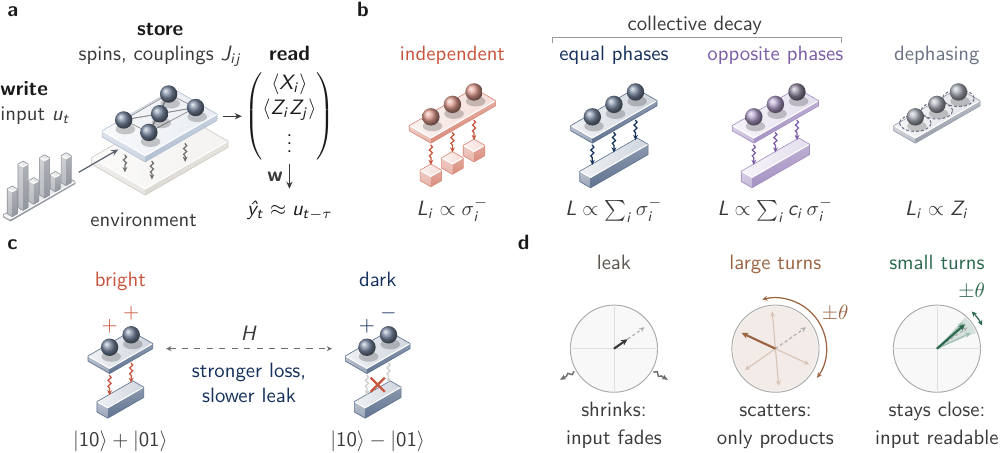}
\caption{\textbf{A quantum reservoir as a memory that is written, stored and read.}
\textbf{a}~Inputs are written into interacting spins, stored in their state and read by a trained linear readout.
\textbf{b}~Four forms of dissipation with the same total coupling, with the rates of Table~\ref{tab:families}. The
darker spins couple with the opposite sign.
\textbf{c}~Two spins that share one jump operator: their decays add (bright) or cancel (dark), and strong loss holds
them in the dark state, the quantum Zeno effect.
\textbf{d}~A stored signal fades because it leaks away or because later inputs turn it at random.}
\label{fig:concept}
\end{figure*}

Processing information over time requires a memory of what came before. A physical system can provide such a memory
because each input changes its state and leaves traces that influence its later response. Reservoir computing uses these
traces to perform a task: an input stream drives a dynamical system, and a weighted sum of quantities describing its
state produces the output. Only these weights are trained; this mapping is called a linear
readout~\cite{Maass2002,JaegerESP2001}. In quantum reservoir computing (QRC), the evolving system is quantum mechanical,
and the readout uses quantities obtained by measuring it~\cite{Fujii2017}. The challenge is to keep earlier inputs
accessible while new information continues to arrive.

The environment shapes how these traces fade. A reservoir must forget its initial state and old inputs, a property
known as fading memory or the echo state property~\cite{Maass2002,JaegerESP2001}, and in a quantum reservoir
dissipation provides it~\cite{ChenNurdin2019}. A system can lose energy or information to its surroundings through
processes collectively described as
dissipation~\cite{ChenNurdin2019,Chen2020,Sannia2024,Ricci2026,Domingo2023,Goetting2025,Mujal2023,Franceschetto2026,KobayashiFeedback2024,Cindrak2024,Sannia2026Memory,Paparelle2026}.
Previous work has shown that tuning the dissipation rate improves memory, in some reservoirs with the best performance
at an intermediate loss rate~\cite{Sannia2024,Goetting2025}. Input, interactions, and dissipation can also be optimized
together~\cite{MartinezPenaPRL2021,Baumann2026}. The rate specifies how strongly the environment acts, but the way it
acts is a separate design choice. Engineering this interaction can prepare and protect selected quantum
states~\cite{Poyatos1996,Verstraete2009,Pastawski2011,Harrington2022}. For a system processing an ongoing input stream,
this raises the possibility of controlling which traces of the past survive.

This distinction becomes concrete for a reservoir of interacting spins, each a two-level quantum system or qubit. When the spins
decay through separate routes into the environment, they undergo independent decay, also called local decay. When they share a common route, they
undergo collective decay. In certain joint states, the different emission paths interfere and cancel, preventing decay
through that shared route. These protected states are called dark
states~\cite{Dicke1954,AsenjoGarcia2017,Facchinetti2016}. Internal dynamics can still move the system into states that
do decay. If this decay is sufficiently rapid, however, it suppresses those transitions, a phenomenon known as the
dissipative quantum Zeno effect~\cite{FacchiPascazio2002,Zanardi2014,Burgarth2019,Reiter2012}. Stronger loss
therefore slows the escape of information from the protected states (Fig.~\ref{fig:concept}c).

A long-lived state is not yet a usable memory, because each new input changes the same state that stores the past.
A signal that initially depends on one input comes to depend on combinations of several inputs, such as their
product. This mixing of inputs is the nonlinearity of the reservoir. It is useful for tasks that require
relationships between inputs, but it reduces the capacity
available for recalling single inputs~\cite{Dambre2012,Inubushi2017,Cindrak2026}. Reading the memory has a cost as
well, because resolving a weak quantum signal requires repeated preparation and
measurement~\cite{Mujal2023,Cindrak2024,Hu2023}, and a long-lived memory is hard to read if its signal is faint.
Previous work has established memory limits in classical reservoirs~\cite{JaegerESP2001,White2004,Ganguli2008},
explored input strength~\cite{Jaeger2007Leaky}, and investigated how noise and protected states affect quantum
information processing~\cite{Kubota2023,Fry2023,Cheamsawat2025,Akshay2026}. Closest to our work,
Ref.~\cite{Sannia2024} tunes the dissipation rate of the reservoir we study, Ref.~\cite{Cindrak2026} finds one
trade-off between memory and nonlinearity across reservoir schemes, and Ref.~\cite{Akshay2026} reads the output of a
reservoir from a decoherence-free subspace to classify quantum states. The Zeno protection of dark states and the
benefit of a weak input for linear memory are established. How protection, continued input, and measurement together
determine the part of the past that a readout recovers is open.

\begin{figure*}[!t]
\centering
\includegraphics[width=\textwidth]{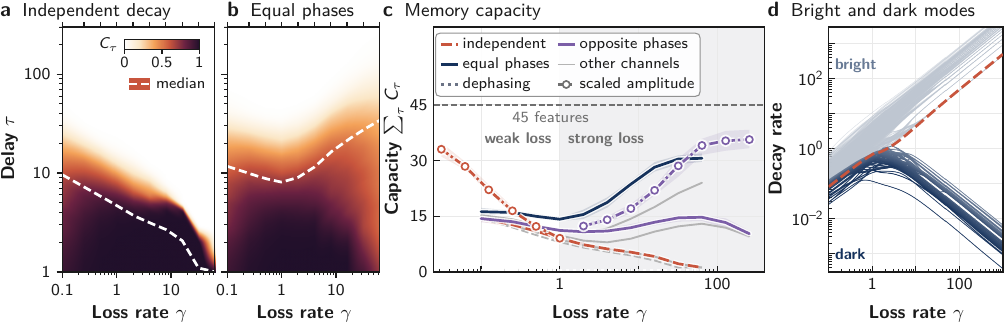}
\caption{\textbf{The decay channel decides whether stronger loss shortens or prolongs the memory.}
\textbf{a,b}~Recall of the input $\tau$ steps back: stronger loss pushes it to shorter delays with independent decay
(a) and to longer delays with the equal-phase channel (b). Half of the recall lies below the dashed line.
\textbf{c}~Linear memory capacity of all channels at a fixed input amplitude, and with the scaled amplitude of
Sec.~\ref{sec:mixing} (circles), in simulated runs summed over the first 300 delays. Shading marks loss stronger
than the couplings ($\gamma>1$).
\textbf{d}~With the equal-phase channel the dark modes decay more slowly as the loss grows, while the slowest mode
of independent decay (dashed) decays faster.}
\label{fig:clock}
\end{figure*}

Here we show that the form of the dissipation matters as much as its rate, because it decides how a memory is
stored, written and read (Fig.~\ref{fig:concept}). Our contributions are threefold.
\begin{enumerate}[label=(\arabic*),leftmargin=*,itemsep=2pt,topsep=3pt]
\item \emph{Storing.} Stronger loss can protect a memory. When the spins decay collectively, the quantum Zeno effect
  holds the stored signal in dark states, and the memory lasts longer the stronger the loss (Sec.~\ref{sec:clock}).
\item \emph{Writing.} A long memory is hard to write, because every new input mixes with the inputs already stored.
  We prove a law for how many inputs mix as a memory grows longer, and we prove that a weaker input, scaled to the
  lifetime, keeps this number fixed. A memory can then last as long as desired and still hold single inputs. In our
  five-spin reservoirs this more than doubles the number of past inputs that can be recalled
  (Sec.~\ref{sec:mixing}).
\item \emph{Reading.} A weaker input leaves a fainter signal. We show that the stationary state decides how much of
  a memory measurements recover, and which form of dissipation recovers the most (Sec.~\ref{sec:price}).
\end{enumerate}
Both theorems hold for classical reservoirs as well. Together, these results make dissipation a design tool. A
reservoir can be built to remember longer under stronger loss, and what it will remember can be computed before it
is built.

%% file: sections/results.tex
\section{Model}
\label{sec:shapes}
The reservoir is the driven spin system of Ref.~\cite{Sannia2024}. Five spins ($N=5$) interact through random
couplings $J_{ij}$, and in each step a random input $u$ between 0 and 1 sets the strength of a transverse field for a
time $\Delta t$:
\begin{equation}
  H(u)=\sum_{i<j}J_{ij}X_iX_j+h\sum_iZ_i+h(1+u)\sum_iX_i ,
  \label{eq:H}
\end{equation}
with Pauli operators $X_i$, $Z_i$ of spin $i$ and a field strength $h$. To weaken the input at a fixed mean drive, we
replace $h(1+u)$ by $h[\tfrac32+\eta(u-\tfrac12)]$. The amplitude $\eta$, known as the input scaling in
reservoir computing~\cite{Jaeger2007Leaky}, sets how far each input moves the field. The state $\rho$ of the spins follows a Lindblad master
equation~\cite{Gorini1976,Lindblad1976},
\begin{equation}
  \dot\rho=-i[H(u),\rho]+\sum_k\Bigl(L_k\rho L_k^\dagger-\tfrac12\{L_k^\dagger L_k,\rho\}\Bigr),
  \label{eq:lindblad}
\end{equation}
where the Hamiltonian drives the internal dynamics and each jump operator $L_k$, scaled by one loss rate $\gamma$, is
one way the environment changes the state, such as removing an excitation or erasing a phase.

A linear readout is trained on 45 features, the expectation values of $X_i$, $Y_i$, $Z_i$ and of the same-axis pairs
$X_iX_j$, $Y_iY_j$, $Z_iZ_j$, to estimate the input from $\tau$ steps back. Its accuracy is the memory
function~\cite{White2004} $C_\tau=\operatorname{corr}^2(\hat y_t,u_{t-\tau})$, the squared correlation between the
estimate $\hat y_t$ and the input on held-out data. We call $C_\tau$ the recall at delay $\tau$. Its sum over the
delays is the linear memory capacity, and one unit is one past input recalled perfectly. The capacities for all
functions of the past inputs add up to the number of features, the total capacity of 45 units~\cite{Dambre2012}.

We call the part of the state that depends on the past inputs the stored signal. The variance that an input leaves in
the state, taken over the inputs that follow, decays in the long run as $(1-1/t_{\mathrm{mem}})^t$, which defines the
memory lifetime $t_{\mathrm{mem}}$ in steps (Appendix~\ref{sec:order}). It depends on the dissipation and on the
inputs, and both theorems use it. Where only the dissipation matters we use the relaxation time $t_{\mathrm{rel}}$,
the decay time of the slow modes at the mean input, which for independent decay is the inverse of the Liouvillian
gap (Table~\ref{tab:times}).

\paragraph{Conventions.}
Capacities quoted in the text come from exact moments, that is, from exact expectation values, the limit of
infinitely many shots, summed over all delays. Figure~\ref{fig:clock}c, the tables and the appendix figures show
simulated runs summed over the first 300 delays. Section~\ref{sec:price} counts shots. All numbers are for five spins with
all-to-all couplings. How the signal scales with the number of spins is open, and reservoirs whose signals
concentrate need exponentially many measurements~\cite{Sannia2025Concentration}.

\paragraph{Decay channels.}
We compare eight dissipators, sets of jump operators that we call decay channels, at equal total coupling $\Tr K$
with $K=\sum_kL_k^\dagger L_k$, and hold the Hamiltonian fixed to isolate the effect of the form of the dissipation.
The Liouvillian, the generator of Eq.~\eqref{eq:lindblad}, has $4^N$ modes, patterns of change in the state that each
decay at their own rate. Equal total coupling fixes their mean rate for every Hamiltonian, but not the slow modes,
which hold a stored signal longest and become slower or faster with the channel
(Appendix~\ref{sec:methods-shapes}, Table~\ref{tab:families}).

We discuss four of them (Fig.~\ref{fig:concept}b). With independent decay every spin loses its excitation on its own.
With collective decay all spins lose it through one shared jump operator, with equal phases or with opposite phases on
spins 4 and 5, and we refer to the equal-phase and the opposite-phase channel. With dephasing every spin loses its
phase. The channels differ in the kernel of the dissipator, the operators that the dissipation alone leaves
unchanged, although the Hamiltonian can still turn them or couple them to decaying ones. For independent decay the
kernel is the ground state alone. A collective channel cannot empty its dark states, so its kernel contains every
operator on the dark subspace. There are $\binom{N}{\lfloor N/2\rfloor}$ dark states, ten for five spins, and their
number grows exponentially with the number of spins. Dephasing leaves the maximally mixed state unchanged whatever
the input, so the driven reservoir relaxes to this state and stores nothing.

\section{Results}
\label{sec:results}

\subsection{Storing: the decay channel sets the memory lifetime}
\label{sec:clock}
Increasing the loss separates the channels into two groups (Fig.~\ref{fig:clock}). Without dark states the memory
shrinks. With independent decay, recall retreats to shorter delays and the linear memory capacity falls steadily
(Fig.~\ref{fig:clock}a,c), and exchange and gain/loss behave in the same way (Fig.~\ref{fig:clock}c, grey). With dark
states the memory survives strong loss. With the equal-phase channel the capacity first dips and then doubles, and
inputs fifty steps back are still recalled (Fig.~\ref{fig:clock}b,c). The decay rates explain the split. Without
dark states the Liouvillian gap, the slowest decay rate, grows in proportion to $\gamma$ and the memory lifetime
shrinks. With dark states the slow rates fall as $1/\gamma$ and the lifetime grows (Fig.~\ref{fig:clock}d). Among
the long-lived channels the capacity grows with the loss for the equal-phase channel and for collective decay in two
groups, and returns to its weak-loss level for the opposite-phase channel and for pairwise decay (Fig.~\ref{fig:clock}c,
grey).

\paragraph{Zeno protection.}
Stronger collective decay prolongs the memory through the quantum Zeno
effect~\cite{Beige2000,FacchiPascazio2002}. In the two-spin picture of Fig.~\ref{fig:concept}c, the dark state
$|{-}\rangle\propto|10\rangle-|01\rangle$ can decay only through the bright state
$|{+}\rangle\propto|10\rangle+|01\rangle$, and strong loss suppresses this escape. Its rate falls as
$\Omega^2/\gamma$, where $\Omega$ couples $|{+}\rangle$ and $|{-}\rangle$, as the slow mode of an overdamped
oscillator settles more slowly the stronger the friction.

For any number of spins the protection follows from perturbation theory in $1/\gamma$ at fixed system size, the
adiabatic elimination of the bright states~\cite{Zanardi2014,Reiter2012}. At strong loss
the reservoir follows its Hamiltonian restricted to the dark states, which turns the stored state without making it
decay. Decay enters at the next order, when the Hamiltonian briefly lifts a dark state into a bright
one. One step is then a rotation plus a loss that falls as $1/\gamma$
(Eq.~\eqref{eq:reduced}), a nearly orthogonal map with a small decay rate, which is the classic recipe for a long
memory in a network~\cite{White2004}. Calculations for three to six spins show the same decay out of the dark states
with different prefactors, and with the equal-phase channel the capacity also grows with the loss at four and six
spins (Appendix~\ref{app:num-lifetime}).

\begin{figure}[!tb]
\centering
\includegraphics[width=\columnwidth]{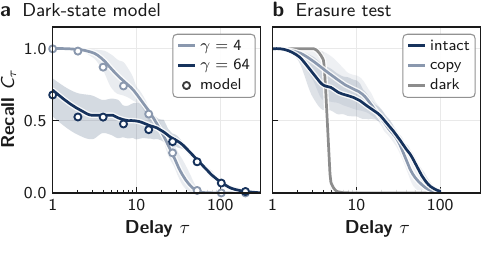}
\caption{\textbf{Dark states hold the memory.} Equal-phase channel.
\textbf{a}~A model of the dark states alone (circles) reproduces the recall of the full reservoir (lines).
\textbf{b}~Erasing what the dark states know about the inputs destroys the memory (dark), while erasing a rotated
copy of them leaves most of it (copy), at $\gamma=16$.}
\label{fig:darkmodel}
\end{figure}

\paragraph{Where the memory is stored.}
The dark states hold the memory, as two tests show (Appendix~\ref{app:num-lifetime}). A model of the ten dark
states alone, with no fitted parameter, reproduces the recall of the full reservoir at strong loss (Fig.~\ref{fig:darkmodel}a). Replacing
the dark-state components of the state by their time averages after every step, which removes their dependence on
the inputs, destroys all memory beyond a few steps, while the same operation on a randomly rotated set of equal size
leaves most of it (Fig.~\ref{fig:darkmodel}b).

\paragraph{Added independent decay.}
Residual independent decay limits the lifetime. For the equal-phase channel, a relaxation time of $T$ steps requires a
ratio of collective to independent decay rates of at least $4.2\,T^2$ (Appendix~\ref{app:num-noise}), so that a ratio
of $10^4$ protects 49 steps. For superconducting qubits coupled to a
waveguide, a four-qubit dark state outlives the decay of a single qubit into the waveguide by more than two orders
of magnitude~\cite{Zanner2022}, and qubit arrays reach a cooperativity of 172~\cite{Mirhosseini2019}.

Dark states protect a stored signal against decay but not against the inputs that follow it. With the opposite-phase
channel the dark states are just as protected as with the equal-phase channel, yet the capacity does not grow.

\subsection{Writing: new inputs mix the stored signal}
\label{sec:mixing}
Each new input turns the stored state by a small angle that depends on the input (Fig.~\ref{fig:concept}d), so every
newer input acts as noise on an older one. Over the memory lifetime the stored direction takes a random walk of small
turns $\pm\theta$ and spreads by about $\theta\sqrt{t_{\mathrm{mem}}}$. Once the spread reaches about a radian, an
old input survives only multiplied by the inputs that followed it, and a linear readout cannot turn such products
back into single inputs. We call this the mixing of inputs. It is the nonlinearity in the trade-off between memory
and nonlinearity~\cite{Dambre2012,Inubushi2017}.

\paragraph{Counting the inputs.}
A stored signal splits into parts that each depend on a fixed set of inputs (Appendix~\ref{sec:order}). The number of
inputs in the set is the order $D$ of the part. A function of $u_{t-3}$ alone has $D=1$, even if it is nonlinear in
$u_{t-3}$, and the product $u_{t-3}u_{t-5}$ has $D=2$. Linear memory requires $D=1$. The total capacity divides among
the orders~\cite{Dambre2012,MartinezPenaIPC2023}, and the share of single inputs is the part with $D=1$.

\paragraph{The mixing law.}
Two quantities decide the order of a stored signal. The memory lifetime $t_{\mathrm{mem}}$ says how long the signal
lasts. The mixing fraction $s$, between 0 and 1, is the part of the signal that one input turns differently from the
average input (Eq.~\eqref{eq:mixing}), and for small turns it grows as the square of the turn. In each step a
fraction $s$ of the stored signal picks up one more input, and a fraction $1/t_{\mathrm{mem}}$ of it is lost. A
signal collects about $s\,t_{\mathrm{mem}}$ inputs before it is lost, and our first theorem makes this exact.

\begin{keybox}
\begin{theorem}[Mixing law, made precise in Appendix~\ref{app:law}]\label{thm:law}
Let a reservoir lose its stored signal slowly: in each step it is reset with a small probability and is otherwise
rotated by the input, under conditions (1) to (3) of Appendix~\ref{app:law}. Pick a part of the stored signal with
probability equal to its share of the total capacity. As the memory lifetime grows, its order $D$ becomes
exponentially distributed, with the mean
\begin{equation*}
  \mathbb ED=s\,t_{\mathrm{mem}}
\end{equation*}
up to an additive constant.
\end{theorem}
\end{keybox}
\begin{figure*}[!t]
\centering
\includegraphics[width=\textwidth]{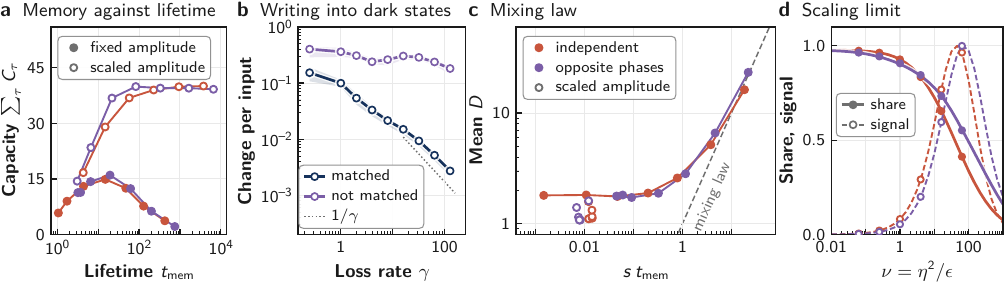}
\caption{\textbf{New inputs mix the stored signal, and a scaled amplitude limits the mixing.}
\textbf{a}~At a fixed input amplitude the linear memory capacity peaks and then falls as the memory grows, and a
scaled amplitude holds it at 40 units.
\textbf{b}~An input that is mode-matched to the equal-phase channel writes less into the dark states as the loss
grows, while an input that is not mode-matched keeps writing.
\textbf{c}~The mean order follows the mixing law, and a scaled amplitude keeps it near one.
\textbf{d}~The scaling limit (lines) against the exact five-spin calculation (markers) for one reservoir, without a
fitted parameter. Solid lines and filled markers show the share of single inputs in the total capacity, dashed lines
and open markers the variance that single inputs leave in the features, scaled to its maximum. Panels a and c show
exact moments, averaged over four reservoirs, two at the longest memories.}
\label{fig:mixing}
\end{figure*}

\paragraph{Scope of the law.}
The proof uses a reservoir with a reset, in any dimension, and nothing in it is quantum. Classical state-affine
reservoirs~\cite{GrigoryevaOrtegaSAS2018} obey the law under the same conditions. For reservoirs without a reset the
same proof applies, with the slow decay in place of the reset, if a single direction of the stored covariance is
slow. Several slow directions keep distinct mixing scales (Appendix~\ref{app:physical}).

\paragraph{Five spins.}
For five spins, Eq.~\eqref{eq:recursion} gives the distribution of $D$ exactly, and we compare it with the law
(Fig.~\ref{fig:mixing}c). At a fixed amplitude, $\eta=1$, the mixing fraction lies between 0.01 and 0.05 wherever
the memory exceeds ten steps. As the memory grows, the mean of $D$ follows $s\,t_{\mathrm{mem}}$ to within about two
inputs, up to 23 inputs at the longest memory, and the distribution of $D$ approaches the exponential law
(Appendix~\ref{app:physical}). The product $s\,t_{\mathrm{mem}}$ is the size of the trade-off between memory and
nonlinearity~\cite{Cindrak2026}, and dissipation enters it through both factors.

At a fixed amplitude a longer memory lowers the linear memory capacity once $s\,t_{\mathrm{mem}}$ exceeds 0.2 to 0.3
(Fig.~\ref{fig:mixing}a). The capacity peaks at 16 units with the opposite-phase channel and at 15 units with
independent decay, and falls to 2 and 4 units at the longest memories (Appendix~\ref{app:num-capacity}). The stored
signal has moved into products of inputs, and with the opposite-phase channel the share of single inputs falls from
0.41 to 0.07. Single inputs stay readable only if the mixing fraction shrinks as the lifetime grows. The phases of
the decay channel and the amplitude of the input both achieve this.

\paragraph{Channel phases.}
The phases of a collective channel decide whether an input reaches the dark states. Spin $i$ couples to the channel
with the sign $c_i=\pm1$, and the collective lowering operator $\sum_ic_i\sigma_i^-$ annihilates the dark states. An
input field is mode-matched to the channel if its strength on spin $i$ is proportional to $c_i$. A mode-matched field $\sum_ic_iX_i$ is the lowering operator plus
its adjoint and cannot couple two dark states, for any number of spins and any Hamiltonian. It writes into them only
through brief excursions into bright states, with a strength that falls as $1/\gamma$. The part of a field that is
not mode-matched writes into them directly (Fig.~\ref{fig:mixing}b, Appendix~\ref{app:num-mixing}).

The uniform input of Eq.~\eqref{eq:H} is mode-matched to the equal-phase channel. Its mixing fraction shrinks as the
loss grows, and this channel reaches twice the capacity of the other two at a fixed amplitude
(Table~\ref{tab:memory}). With the opposite-phase channel the same input is not mode-matched. It writes directly into
the dark states and mixes them at every loss rate. The two channels differ in nothing else. Rotating spins 4 and 5
by $Z$ maps the opposite-phase channel with the uniform input onto the equal-phase channel with an input of opposite
sign on these two spins, and changes only the signs of couplings $J_{ij}$, whose distribution is symmetric. The
missing mode matching is the generic case, because a phase error opens a direct coupling within the dark subspace
(Appendix~\ref{app:num-mixing}).

\paragraph{Input scaling.}
The input amplitude $\eta$ controls the mixing for every channel. A weak input is known to favour linear
memory~\cite{Jaeger2007Leaky,Inubushi2017,Cindrak2026}, because it suppresses the products of the latest inputs that
the reservoir forms within a few steps. The mixing law adds what a long memory requires. The mixing fraction grows
as $\eta^2$, and the lifetime as the inverse of the leak rate $\epsilon$ at which the stored signal decays. The
product $s\,t_{\mathrm{mem}}$ stays fixed only if the ratio $\nu=\eta^2/\epsilon$ does. We call an amplitude with fixed $\nu$, which falls as $1/\sqrt{t_{\mathrm{mem}}}$, a scaled amplitude, and
$\eta=1$ the fixed amplitude. Theorem~\ref{thm:law} shows that the scaling is necessary, because at a fixed amplitude
the share of single inputs vanishes as the memory grows. Our second theorem shows that it is sufficient.

\begin{keybox}
\begin{theorem}[Scaling limit, made precise in Appendix~\ref{app:gentle}]\label{thm:gentle}
Let one step of a reservoir be a rotation, a leak at a small rate $\epsilon$ and an input-dependent turn of amplitude
$\eta$, under the decay condition of Appendix~\ref{app:gentle}. Let the leak vanish while $\nu$ stays fixed. Then the
memory lifetime grows without limit, while the distribution of the order $D$ converges. Its mean stays finite, and
single inputs keep a fixed share of the total capacity. Both depend on the input only through $\nu$. For large $\nu$
the mean of $D$ grows in proportion to $\nu$, as the mixing law requires.
\end{theorem}
\end{keybox}

The leak rate is $\epsilon=\gamma$ for independent decay and $\epsilon=1/\gamma$ for the dark states, and the limit
follows from the Hamiltonian and the jump operators without a fitted parameter. One qubit gives a closed form
(Appendix~\ref{app:qubit}). For three spins with the opposite-phase channel, interval arithmetic confirms both
theorems at strong loss. At a fixed amplitude single inputs lose their share of the capacity, and with a scaled
amplitude they keep it. The stored signal resides there in coherences between the energy levels of the dark states
(Appendix~\ref{app:spins}).

\paragraph{Comparison with five spins.}
For five spins we use $\nu=1$, an amplitude of $1/\sqrt\gamma$ for the opposite-phase channel and of $\sqrt\gamma$
for independent decay. The limit reproduces the exact five-spin calculation over three decades of $\nu$, with the
share of single inputs agreeing to within 0.01 (Fig.~\ref{fig:mixing}d, Appendix~\ref{app:num-capacity}). At $\nu=1$
single inputs carry about 0.9 of the total capacity, and the linear memory capacity settles at 40 units for both
channels as the lifetime grows (Fig.~\ref{fig:mixing}a). At the best loss rate of a fixed amplitude it is 15 units
with independent decay and 16 with the opposite-phase channel. Simulated runs on eight reservoirs, with a readout
trained on a finite number of inputs and summed over the first 300 delays, reach 35 units (Fig.~\ref{fig:clock}c,
Table~\ref{tab:memory}).

The total capacity is fixed, so the gain moves capacity from combinations of inputs to single inputs. A task that needs products of inputs calls for a larger amplitude, and
$s\,t_{\mathrm{mem}}$ gives the number of inputs that a stored signal then combines. A long memory needs no dark
states, because weak independent decay with a scaled amplitude reaches one too. Dark states add a memory that stays
long under strong dissipation and grows with the loss.

\paragraph{The cost of a weak input.}
The scaled amplitude is an upper bound on the input and not an optimum. With exact expectation values the capacity
keeps rising as the input weakens, towards the number of features (Appendix~\ref{app:num-amplitude}), because a
weaker input mixes less and a fainter signal costs nothing without measurement noise. Theorem~\ref{thm:gentle} quantifies what a weak input gives up
(Fig.~\ref{fig:mixing}d). The variance that single inputs leave in the features vanishes in proportion to $\nu$ for
weak inputs. It is largest where a stored signal combines a few inputs and single inputs carry about half of the
total capacity (Appendix~\ref{app:num-capacity}).

\begin{figure*}[!t]
\centering
\includegraphics[width=\textwidth]{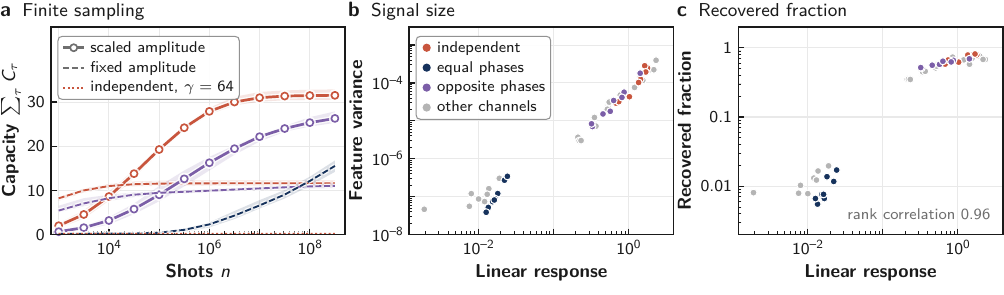}
\caption{\textbf{Finite sampling limits the memory that can be read out.}
\textbf{a}~Capacity recovered against the number of shots per setting and step with the rewinding protocol, for
independent decay at weak loss ($\gamma=1/32$) and for the opposite-phase and the equal-phase channel at strong loss
($\gamma=64$).
\textbf{b}~The variance of the features grows with the linear response of the stationary state to the input field.
\textbf{c}~So does the fraction of the exact capacity that $10^6$ shots recover.
Points in b and c are single reservoirs at $\eta=0.35$. Their loss rates give every channel the same relaxation
time, $t_{\mathrm{rel}}=32$ steps, and every shot repeats the full input history.}
\label{fig:sampling}
\end{figure*}

\subsection{Reading: the stationary state sets the measurable signal}
\label{sec:price}
A weaker input leaves a fainter signal. This costs nothing with exact expectation values and many measurements
otherwise. Each feature is estimated from $n$ shots per measurement setting and step, and its statistical error must
stay below the signal, the variation of the feature with the inputs (Appendix~\ref{sec:methods-sampling}).

\paragraph{Signal size.}
The stationary state sets the size of the signal. An input acts on the stationary state $\rho_*$, the steady state
of the dynamics averaged over the inputs, through the commutator with the input field, and only the part of this change that lies on the
slow modes is remembered. Its size is the linear response of the stationary state to the input field
(Appendix~\ref{app:num-measurement}), in the spirit of the link between memory and optical absorption of
Ref.~\cite{Goetting2025}.

To compare all eight channels we choose the loss rate of each so that its relaxation time is 32 steps, and set
$\eta=0.35$. The relaxation time replaces the memory lifetime here because it does not depend on the input. The
variance of the features then follows the linear response over four orders of magnitude for the seven channels that
respond (Fig.~\ref{fig:sampling}b). The fraction of the exact capacity that $10^6$ shots recover follows it as well
(Fig.~\ref{fig:sampling}c). The rank correlation is 0.96 over all reservoirs and at least 0.81 within six of the
seven channels (Appendix~\ref{app:num-measurement}). The equal-phase channel leaves a polarized state that the
uniform input moves only through bright states. It has the largest exact capacity of all channels, and $10^6$ shots
recover less than 1\% of it. Reversing the coupling phase of two spins lets the input move the state directly, and
the opposite-phase channel recovers 60\% of its exact capacity, against 73\% for independent decay. The exact
capacity does not predict what measurements recover. The linear response predicts which fraction of it they do.

\paragraph{Strong loss.}
The form of the dissipation matters most where the loss is strong (Fig.~\ref{fig:sampling}a). At $\gamma=64$ and a
fixed amplitude, $10^4$ shots recover 8 units with the opposite-phase channel and nothing above the noise floor with
the equal-phase channel or with independent decay (Table~\ref{tab:memory}). A collective channel that is not
mode-matched to the input is the right choice when strong loss is given, as for emitters that radiate into a common
mode~\cite{Dicke1954,AsenjoGarcia2017,Facchinetti2016}. Where the loss can be made weak, independent decay recovers
more, 11 units from the same number of shots. A short memory needs few shots, and the scaled amplitude pays off only
with more. Its capacity overtakes that of the fixed amplitude between $10^4$ and $10^5$ shots
(Appendix~\ref{app:num-measurement}).

\paragraph{Total cost.}
A measurement destroys the state, so every shot restarts the reservoir and repeats the recent inputs before the
measurement, the rewinding protocol~\cite{Mujal2023,Cindrak2024}. The number of repeated steps grows with the memory
lifetime. When every repeated step is counted (Appendix~\ref{sec:methods-sampling}), a budget of $10^9$ reservoir
steps per output step gives the three channels the same capacity at a fixed amplitude, each at its best loss rate
(Table~\ref{tab:memory}). The opposite-phase channel holds at a hundredfold stronger loss what independent decay
holds at weak loss. With the scaled amplitude, weak independent decay recovers the most, 34 units against 21 with
the opposite-phase channel.

With a finite number of shots a best amplitude exists. It falls as a power of the relaxation time, with fitted
exponents between $-0.40$ and $-0.46$, close to the $-1/2$ of the scaled amplitude
(Appendix~\ref{app:num-amplitude}). The lifetime bounds the input amplitude from above, and the measurement budget
bounds it from below.

%% file: sections/discussion.tex
\section{Conclusion}
\label{sec:conclusion}
Until now, the usual control over the dissipation of a quantum reservoir was its
rate~\cite{Sannia2024,Goetting2025}, and three things were unknown. It was unknown whether a reservoir can keep its
memory under strong loss. It was unknown how much of a long memory survives the inputs that keep arriving, because
the trade-off between memory and nonlinearity had no law. And it was unknown which form of dissipation leaves a
memory that measurements can read.

All three are now known. Strong loss can protect a memory. With collective decay the memory is held in dark states
and lasts longer the stronger the loss. The trade-off between memory and nonlinearity has a law. The longer a memory
lasts, the more inputs mix into it, and a weaker input, scaled to the lifetime, stops this growth, so a memory can
last as long as desired and still hold single inputs. Both statements are theorems. They hold for classical and
quantum reservoirs alike, and they agree with our simulations without a fitted parameter. Finally, the stationary
state decides how much of a memory measurements can read, and with it which dissipation to build.

For the field, this turns dissipation from a rate to be tuned into a structure to be designed. Platforms with
collective loss, such as emitters that share a waveguide or a cavity, can use strong loss to lengthen a memory
instead of fighting it. The trade-off between memory and nonlinearity can be computed from the equations of motion,
before a reservoir is built or trained, and for classical reservoirs as well. Three design rules follow. Store a
memory in dark states, write it with an input that reaches them and is scaled to the lifetime, and read it where the
stationary state responds to the input. Engineered dissipation is known to prepare and protect quantum
states~\cite{Poyatos1996,Verstraete2009}, and it now also shapes a memory. Dissipation makes a reservoir forget, and
its form decides what the reservoir remembers.

%% file: sections/methods.tex
\section{Methods}
\label{sec:methods}
This appendix specifies the five-spin simulations. Table~\ref{tab:times} lists the time scales.

\begin{table}[!t]
\caption{\label{tab:times}Time scales, in input steps of length $\Delta t$.}
\centering\footnotesize\setlength{\tabcolsep}{4pt}
\begin{tabular}{@{}lp{0.80\columnwidth}@{}}
\toprule
$t_{\mathrm{mem}}$ & memory lifetime: decay time of the variance of the stored signal (Appendix~\ref{sec:order}). It
  depends on the inputs and enters both theorems.\\
$t_{\mathrm{rel}}$ & relaxation time: decay time of the slow modes of the Liouvillian at the mean input
  (Appendix~\ref{sec:methods-slow}). It does not depend on the inputs. Used in Secs.~\ref{sec:clock}
  and~\ref{sec:price}.\\
$t_{\mathrm f}$ & recall time: decay time of the response to one input, with
  $1/t_{\mathrm f}=1/t_{\mathrm{rel}}+O(\eta^2)$ (Appendix~\ref{sec:methods-slow}). Used in
  Appendix~\ref{app:numerics} only.\\
$\tau_{90}$ & rewinding window: the first $\tau_{90}$ delays hold 90\% of the linear memory capacity
  (Appendix~\ref{sec:methods-sampling})\\
\bottomrule
\end{tabular}
\end{table}

\subsection{Reservoir, readout and capacity}
\label{sec:methods-reservoir}
The Hamiltonian of Eq.~\eqref{eq:H} uses $h=1/2$ and couplings $J_{ij}$ drawn uniformly from $[-1,1]$, so the
transverse field runs from $h$ to $2h$. Each input $u$ is drawn uniformly from $[0,1]$ and held for a step of length
$\Delta t=1/2$. We integrate Eq.~\eqref{eq:lindblad} in the 1024-dimensional basis of Pauli operators and evaluate the
step map $E(u)=e^{\Delta t\mathcal L(u)}$ by Chebyshev interpolation in $u$ at ten nodes. Each reservoir draws its
couplings and its input stream from one seed, and every channel uses the same eight seeds. After a washout of $\max(800,60\gamma)$ steps, at least four times the recall
time, we record the 45 features, the expectation values of $X_i$, $Y_i$, $Z_i$ and of $X_iX_j$, $Y_iY_j$, $Z_iZ_j$, for
2000 training, 400 validation and 1600 test inputs. The readout is a least-squares fit with an intercept on the
training and validation inputs, and $C_\tau$ is the squared correlation of its estimate with $u_{t-\tau}$ on the test
inputs. In these simulated runs the linear memory capacity sums $C_\tau$ over delays 1 to 300, and the exact
moments of Appendix~\ref{sec:order} give the sum over all delays. A delay without any recall reads about $1/1600$
on the test inputs, so the sum has a floor of about 0.2 units. Unless a caption states otherwise, figures show
the mean of the eight reservoirs, with bands from the 16th to the 84th percentile.

\subsection{Decay channels and equal coupling}
\label{sec:methods-shapes}
Table~\ref{tab:families} lists the jump operators of the eight decay channels. Their weights give every channel the
same total coupling $\Tr K=N2^{N-1}\gamma=80\gamma$, with $K=\sum_kL_k^\dagger L_k$, which fixes the mean decay rate. For a traceless jump $L$ the
superoperator trace of its dissipator is $-2^N\Tr L^\dagger L$, while the commutator with any Hamiltonian is traceless,
so the $4^N$ eigenvalues $\lambda$ of the Liouvillian obey
\begin{equation}
  \frac{1}{4^N}\sum\operatorname{Re}\lambda=-\frac{\Tr K}{2^N}
  \label{eq:trace}
\end{equation}
for every channel and every Hamiltonian. The kernel of the dissipator is larger than the set of dark states: the ten
dark states of a collective channel carry $10\times10=100$ operators.

\begin{table}[!t]
\caption{\label{tab:families}Decay channels at five spins, each with $\Tr K=80\gamma$, and the number of operators
that the dissipator leaves unchanged. The operators $\sigma_i^\mp$ lower and raise spin $i$, with $i<j$ for pairwise
decay, $i\neq j$ for exchange and $c=(1,1,1,-1,-1)$.}
\centering\footnotesize\setlength{\tabcolsep}{3.4pt}
\begin{tabular}{lll}
\toprule
channel & jump operators & \begin{tabular}[b]{@{}l@{}}kernel\\dimension\end{tabular}\\\midrule
independent & $\sqrt\gamma\,\sigma_i^-$ & 1\\
equal phases & $\sqrt\gamma\sum_i\sigma_i^-$ & 100\\
opposite phases & $\sqrt\gamma\sum_ic_i\sigma_i^-$ & 100\\
two groups & $\sqrt\gamma\sum_{i\le3}\sigma_i^-$, $\sqrt\gamma\sum_{i>3}\sigma_i^-$ & 36\\
pairwise & $\sqrt\gamma\,\sigma_i^-\sigma_j^-$ & 36\\
exchange & $\sqrt{0.4\gamma}\,\sigma_i^-$, $\sqrt{0.3\gamma}\,\sigma_i^-\sigma_j^+$ & 1\\
gain/loss & $\sqrt{\gamma/1.3}\,\sigma_i^-$, $\sqrt{0.3\gamma/1.3}\,\sigma_i^+$ & 1\\
dephasing & $\sqrt{\gamma/2}\,Z_i$ & 32\\
\bottomrule
\end{tabular}
\end{table}

\subsection{Slow modes, time scales and the effective Liouvillian}
\label{sec:methods-slow}
The input-averaged step $\bar E=\mathbb E_uE(u)$ is computed by six-point Gauss--Legendre quadrature. Its eigenvalues
$e_j$ give decay rates $\kappa_j=-\ln|e_j|/\Delta t$. If the kernel of the dissipator has dimension $d_0$, the slow
modes are the $d_0-1$ slowest decaying ones besides the stationary state, and for $d_0=1$ the slowest mode. With
their mean rate $\bar\kappa$ the recall time is $t_{\mathrm f}=1/(\Delta t\,\bar\kappa)$. The averaged step carries
the part of a stored signal that still depends on a single input. Its mean rate adds decay and mixing to second
order in the amplitude, $\bar\kappa=\bar\kappa_0+c\,\eta^2+O(\eta^4)$, because the inputs are distributed
symmetrically about their mean and the term of $\bar E$ linear in $\eta$ vanishes. Appendix~\ref{app:num-mixing}
tests this sum numerically. The relaxation time is the same quantity without the mixing,
$t_{\mathrm{rel}}=1/(\Delta t\,\bar\kappa_0)$, where $\bar\kappa_0$ is the mean decay rate of the slow modes of the
Liouvillian at the mean input, $u=1/2$. Figure~\ref{fig:clock}d shows the decay rates $-\operatorname{Re}\lambda$ of
this Liouvillian. With a collective channel we call the $d_0-1$ slowest of them dark and all faster ones bright.

Let $\mathcal D$ be the unit-rate dissipator, $\mathcal P=\lim_{t\to\infty}e^{t\mathcal D}$ the projector onto its
kernel and $\mathcal S=(\mathcal D+\mathcal P)^{-1}-\mathcal P$ its inverse on the decaying operators. At strong loss,
adiabatic elimination of the decaying operators to second order~\cite{Zanardi2014,Reiter2012} gives the effective
Liouvillian on the kernel,
\begin{equation}
  \mathcal G(u)=\mathcal P\mathcal H(u)\mathcal P-\frac1\gamma\,\mathcal P\mathcal H(u)\mathcal S\mathcal H(u)\mathcal P ,
  \label{eq:reduced}
\end{equation}
with $\mathcal H=-i[H,\cdot\,]$. The first term is the Hamiltonian restricted to the kernel, which turns its
operators, and the second is the loss through brief excursions into bright states.

For five spins the scaled amplitude has $\nu=1$, which is $\eta=\min(1,\gamma^{-1/2})$ for the opposite-phase
channel and $\eta=\min(1,\gamma^{1/2})$ for independent decay.

\subsection{Number of inputs mixed}
\label{sec:order}
The order $D$ of a part of the stored signal is the number of distinct inputs it depends on, the current input
included (Sec.~\ref{sec:mixing}).
Precisely, for independent inputs, every function of the past inputs with finite variance splits uniquely into orthogonal parts, each of which
depends on a fixed finite set of inputs, the Hoeffding decomposition~\cite{Hoeffding1948}. The order of a part is the
size of its set, not a polynomial degree. For signs
$u_t=\pm1$ the parts are the products $\prod_{j\in J}u_{t-j}$, and for continuous inputs one uses orthogonal
polynomials in each input. We collect the traceless
coefficients of the state in an orthonormal operator basis in a vector $r$, on which quantum channels act
affinely~\cite{MartinezPenaOrtega2023,MartinezPenaOrtega2025}, $r_t=A(u_t)r_{t-1}+b(u_t)$. All means and covariances
are taken over input histories. Let $\bar A=\mathbb EA$
be the mean step, $m$ the stationary mean and $\xi=A(u)m+b(u)-m$ the innovation, the change that one input writes,
with covariance $V$. With $\mathcal BX=\bar AX\bar A^{\mathsf T}$, $\mathcal TX=\mathbb E[AXA^{\mathsf T}]$ and
$\mathcal K=\mathcal T-\mathcal B$, the covariance $G_k$ carried by products of order $k$ obeys
\begin{equation}
  G_1=(1-\mathcal B)^{-1}V,\qquad G_{k+1}=(1-\mathcal B)^{-1}\mathcal KG_k .
  \label{eq:recursion}
\end{equation}
The innovation adds the current input alone. The mean step carries a stored product along without adding an input. The
fluctuating part $\mathcal K$ of the step multiplies it by the current input, which raises its order by one. The $G_k$
add up to the stationary covariance $\Sigma$. A readout of features $y=Mr$ has the capacity
$\Tr[(M\Sigma M^{\mathsf T})^+MG_kM^{\mathsf T}]$ of order $k$, and these capacities add up to the rank of the
readout~\cite{Dambre2012}. Divided by their sum, they give the share of each order, the distribution of $D$ in
Theorem~\ref{thm:law}. For one observable with coordinates $a$ the share of order $k$ is
$a^{\mathsf T}G_ka/a^{\mathsf T}\Sigma a$. The share of single inputs is the share of order one, which includes the
current input. Without the current input it is $a^{\mathsf T}\bar AG_1\bar A^{\mathsf T}a/a^{\mathsf T}\Sigma a$, the
share of single past inputs. It is smaller by at most 0.08 in the five-spin results and is the quantity bounded for
three spins in Appendix~\ref{app:spins}.

The memory lifetime and the mixing fraction come from the slowest mode of the covariance step $\mathcal T$: its
eigenvalue is $1-1/t_{\mathrm{mem}}$, and with its right eigenvector $Q$ and left eigenvector $\ell$ the mixing fraction
is $s=\ell(\mathcal KQ)/\ell(Q)$, as in Appendix~\ref{app:physical}. For five spins we evaluate
Eq.~\eqref{eq:recursion} on the 1023 traceless Pauli coordinates (Appendix~\ref{app:num-capacity}).

\subsection{Finite sampling and the measurement budget}
\label{sec:methods-sampling}
A quantum measurement gives a random outcome, so every feature is an average over $n$ shots per measurement
setting and step. We use three settings, with all spins read along $X$, $Y$ or $Z$. For each step we form the outcome
distributions of the three settings and draw $n$ shots from each. The readout is fitted as above, with a ridge
chosen from fourteen values between 0 and $10^4$ on the validation inputs. In Sec.~\ref{sec:price} the distributions
come from the rewinding protocol~\cite{Mujal2023}: for each output step the reservoir is restarted in
$|00000\rangle$ and the last $\tau_{90}$ inputs are repeated, where the first $\tau_{90}$ delays hold 90\% of the exact
linear memory capacity. We count the budget $B$ in reservoir steps, shots times repeated steps summed over the three settings, per
output step, so that $n=B/(3\tau_{90})$, and we also run $n=10^3$ to $10^{8.5}$ in half decades.
Figure~\ref{fig:sampling}b,c, Table~\ref{tab:matched} and the amplitude map of Appendix~\ref{app:num-amplitude} use
the state with the full history instead. For the first two the loss rates are matched to a relaxation time of 32
steps (Appendix~\ref{app:num-measurement}).

%% file: sections/appendix.tex
\section{The mixing law and the scaling limit}
\label{app:law}
This appendix proves the two theorems and extends the mixing law to reservoirs without a reset.

\subsection{Setting and statement}
We describe the reservoir by the $d$ traceless coordinates $r$ of its state in an orthonormal basis of operators, the
Bloch vector for a qubit. A unitary acts on $r$ as a rotation. In each step the reservoir is reset to a fixed state
with a small probability $p$ and is otherwise rotated by $O(u)$:
\begin{equation}
  r_t=(1-p)\,O(u_t)\,r_{t-1}+p\,b .
  \label{eq:reset}
\end{equation}
Here $b$ represents the reset state. Equation~\eqref{eq:reset} averages over whether the reset occurred, so the only
remaining randomness is the input history, with independent and identically distributed inputs $u_t$. A stored signal loses a fraction $1-(1-p)^2=2p-p^2$ of its variance per step, so its memory lifetime, defined by
the slow eigenvalue $1-1/t_{\mathrm{mem}}$ of the covariance step (Appendix~\ref{sec:order}), is
$t_{\mathrm{mem}}=1/(2p-p^2)$, or $1/(2p)$ to leading order. Three
conditions make the law hold:
\begin{enumerate}[label=(\arabic*),leftmargin=*,itemsep=2pt,topsep=4pt]
\item \emph{The average step forgets.} The average rotation $A=\mathbb EO$ shrinks every vector: its spectral radius
  is below one.
\item \emph{The inputs mix.} The only covariance matrices that the random rotations leave unchanged on average,
  $\mathbb E[OXO^{\mathsf T}]=X$, are multiples of the identity.
\item \emph{The input writes.} One input changes the state: $\mathbb E\|(O-A)\mu\|^2>0$ with $\mu=(1-A)^{-1}b$.
\end{enumerate}
The first condition means that only products of inputs can be stored for long, the second that a stored signal
spreads evenly over all directions, and the third that there is something to store. The mixing fraction
\begin{equation}
  s=1-\frac{\|A\|_F^2}{d}
  \label{eq:mixing}
\end{equation}
says how strongly one input turns the stored state. For a unit vector $x$ of random orientation,
$\mathbb E\|(O-A)x\|^2=s$: one input turns a fraction $s$ of a stored signal away from the average rotation. The
fraction lies between 0 and 1, the first condition makes it positive, and for small turns $\theta$ it grows as
$\theta^2$.

\begin{keybox}
\textbf{Mixing law, precise form.} Assume conditions (1)--(3) and take any fixed linear readout. Pick a part of
the stored signal with probability equal to its share of the readout's capacity, and let $D$ be its order, the number
of distinct inputs it depends on. As $p\to0$, the ratio
$D/t_{\mathrm{mem}}$ becomes exponentially distributed with mean $s$, and
\[\mathbb ED=s\,t_{\mathrm{mem}}+O(1).\]
Equivalently, products of more than $x\,t_{\mathrm{mem}}$ inputs carry the share $e^{-x/s}$ of the capacity.
\end{keybox}

\subsection{Proof}
\label{app:proof}
The proof makes the competition of Sec.~\ref{sec:mixing} exact: on average $1/s$ steps pass before the next input joins a
product, which meanwhile loses a fraction $1/(s\,t_{\mathrm{mem}})$ of its weight.
The proof works with covariance matrices. Appendix~\ref{sec:order} splits the stationary covariance $\Sigma$ of the state into the
parts $G_k$ carried by products of $k$ inputs and gives their recursion, Eq.~\eqref{eq:recursion}: the mean step
$\mathcal B$ carries a stored product along, and the fluctuating part $\mathcal K$ of the step adds one input to it.
In the reset model the mean step is $\mathcal BX=(1-p)^2AXA^{\mathsf T}$, the full step
$\mathcal TX=(1-p)^2\,\mathbb E[OXO^{\mathsf T}]$, and $\mathcal K=\mathcal T-\mathcal B$. We write $\mathcal B_0$,
$\mathcal T_0$ and $\mathcal K_0$ for the same maps at $p=0$.

\paragraph{Step 1: the stored signal is small and spread evenly.}
The stationary mean is $m=p\,[1-(1-p)A]^{-1}b=p\mu+O(p^2)$, so one input writes the change $(1-p)(O-A)m$, of size
$p$. Its covariance $V$ has the trace $p^2v+O(p^3)$ with $v=\mathbb E\|(O-A)\mu\|^2$, which is positive by
condition (3). The random rotations preserve the trace of a covariance and, by condition (2), leave only the identity unchanged. Since they cannot enlarge a covariance, the identity is a simple fixed point with the projector
$X\mapsto(\Tr X/d)\,1$. On every other direction $(1-\mathcal T)^{-1}$ stays bounded as $p\to0$, while on the identity
it equals $1/[1-(1-p)^2]\approx1/(2p)$. The written covariance therefore piles up along the identity:
\begin{equation}
  \Sigma=(1-\mathcal T)^{-1}V=\frac{pv}{2d}\,1+O(p^2).
  \label{eq:slowcov}
\end{equation}

\paragraph{Step 2: few inputs carry a vanishing share.}
The mean step forgets by condition (1), so $(1-\mathcal B)^{-1}$ stays bounded as $p\to0$. The recursion then makes
$G_1=(1-\mathcal B)^{-1}V$ of order $p^2$, and the same holds for every $G_k$ with fixed $k$. Measured against
$\Sigma$, which is of order $p$ and has full rank, any fixed set of orders carries a share of order $p$ of the
capacity of any readout. A long memory consists of products of many
inputs.

\paragraph{Step 3: counting the inputs.}
Call $\mathcal F=(1-\mathcal B)^{-1}\mathcal K$ the order transfer, so that $G_{k+1}=\mathcal FG_k$. At $p=0$ it
leaves the identity unchanged, because $\mathcal K_01=1-AA^{\mathsf T}=(1-\mathcal B_0)1$. Its matching left
eigenvector is the weight
\begin{equation}
  \ell(X)=\Tr[(1-A^{\mathsf T}A)X],\qquad \ell(1)=ds ,
  \label{eq:weight}
\end{equation}
since $\ell(\mathcal F_0X)=\Tr\mathcal K_0X=\ell(X)$ when the rotations preserve the trace. For small $p$ the
eigenvalue one moves to
\begin{equation}
  1-\frac{2}{s}\,p+O(p^2):
  \label{eq:ordereig}
\end{equation}
each step up in order costs a fraction $2p/s=1/(s\,t_{\mathrm{mem}})$ of the weight, as in that competition. To see this,
differentiate $(1-\mathcal B)\mathcal F=\mathcal K$ at $p=0$ and apply it to the identity, which gives
$(1-\mathcal B_0)\mathcal F_0'1=-2\cdot1$. Its trace is the weight $\ell(\mathcal F_0'1)=-2d$, and dividing by
$\ell(1)=ds$ gives the rate $2/s$.

To turn the eigenvalue into a distribution, weigh order $k$ with $e^{-\zeta pk}$. Summing the recursion gives
$\sum_kz^kG_k=z(1-z\mathcal F)^{-1}G_1$, and as $z=e^{-\zeta p}$ approaches one, only the eigenvalue of
Eq.~\eqref{eq:ordereig} makes the inverse large. With $\ell(G_1)=p^2v+O(p^3)$ and $\beta=2/s$,
\begin{equation}
  \sum_{k\ge1}e^{-\zeta pk}G_k=\frac{pv}{2d}\,\frac{\beta}{\zeta+\beta}\,1+O(p^2).
  \label{eq:laplace}
\end{equation}
Divided by Eq.~\eqref{eq:slowcov}, this is $\beta/(\zeta+\beta)$, the Laplace transform of an exponential
distribution of $pD$ with rate $\beta$. So $D/t_{\mathrm{mem}}=2pD\,[1+O(p)]$ is exponential with mean $s$ in the limit. Because both
covariances are multiples of the identity to leading order, the same holds for the share of the capacity of any
readout. Finally, $\sum_kkG_k=(1-\mathcal F)^{-1}\Sigma=\Sigma/(\beta p)+O(p)$ gives
$\mathbb ED=s\,t_{\mathrm{mem}}+O(1)$.

\paragraph{Two technical points.}
On symmetric matrices, where covariances live, the eigenvalue one of $\mathcal F_0$ is simple: a fixed point
satisfies $\mathcal T_0X=X$ and is
therefore a multiple of the identity, and $\ell(1)>0$ rules out a Jordan chain. For $p>0$ every eigenvalue of
$\mathcal F$ lies inside the unit circle, by a Perron--Frobenius argument: $\mathcal F$ maps positive matrices to
positive matrices, and so does $(1-\mathcal F)^{-1}=1+(1-\mathcal T)^{-1}\mathcal K$. The eigenvalues of
$\mathcal F_0$ therefore lie in the closed unit disk, only the simple one equals one, and only the eigenvalue of
Eq.~\eqref{eq:ordereig} makes the sum diverge.

\subsection{Reservoirs without a reset}
\label{app:physical}
A physical reservoir does not reset to a prescribed state, but the proof used the reset only through three facts: one
covariance direction is slow, it leaks at a small rate, and the mean step forgets everything else. The same three
steps therefore prove the law for general reservoirs.

\emph{Setting.} One step is $r_t=A_u(\epsilon)r_{t-1}+b_u(\epsilon)$ on the coordinates that the inputs can reach,
with a small leak $\epsilon$, and $\mathcal B$, $\mathcal T$ and $\mathcal K$ are formed from $A_u$ as before. We
assume:
\begin{enumerate}[label=(\roman*),leftmargin=*,itemsep=2pt,topsep=4pt]
\item $A_u$ and $b_u$ depend smoothly on $\epsilon$ in mean square, and $b_u$ vanishes at $\epsilon=0$.
\item The mean step forgets at $\epsilon=0$: $\mathbb EA_u(0)$ has spectral radius below one.
\item The reservoir is mean-square stable for small $\epsilon>0$: the covariance step $\mathcal T$ has spectral radius
  below one.
\item One covariance direction is slow: $\mathcal T$ has a simple eigenvalue $1-\kappa\epsilon+O(\epsilon^2)$ with
  $\kappa>0$ and no other eigenvalue near one. At $\epsilon=0$ its right eigenvector is a positive matrix $Q$, and its
  left eigenvector is a weight $\ell$ that is nonnegative on positive matrices, with $\ell(Q)=1$.
\item The input writes along it: the change written by one input, of order $\epsilon$, has covariance
  $\epsilon^2V_*+O(\epsilon^3)$ with $\ell(V_*)>0$.
\end{enumerate}
The reset model is the case $\epsilon=p$, $\kappa=2$, $Q=1$ and $\ell=\Tr/d$. Condition (iv) concerns the
covariance that the random steps carry, not the number of slow modes: at a fixed amplitude the inputs turn the 99
slow dark modes of the opposite-phase channel differently in each step and spread a stored signal over all of them
faster than it leaks, so that one covariance direction stays slow. Without (iv) one observable shows several order
scales: two reset sectors with mixing fractions $1/4$ and $3/4$, selected by a classical flag, give products of more
than $x\,t_{\mathrm{mem}}$ inputs the share $w\,e^{-4x}+(1-w)\,e^{-4x/3}$ with $w=0.82$, which exact moments
reproduce.

\begin{keybox}
\textbf{Mixing law without a reset.} Assume (i)--(v) and take a fixed readout $y=Mr$ whose projected slow
covariance $MQM^{\mathsf T}$ has full rank. Then the mixing fraction $s=\ell(\mathcal K_0Q)$ is positive and at most one, and $\epsilon D$ becomes exponentially distributed with rate $\kappa/s$. The mean order is
$s/(\kappa\epsilon)+O(1)$: the mixing fraction times the memory lifetime $1/(\kappa\epsilon)$.
\end{keybox}

\emph{Proof.} Steps 1 and 2 carry over. The stationary mean is of order $\epsilon$ by (i) and (ii), and the slow
eigenvalue piles the written covariance up along $Q$,
\begin{equation*}
  \Sigma=\epsilon\,\frac{\ell(V_*)}{\kappa}\,Q+O(\epsilon^2),
\end{equation*}
and every fixed order carries a covariance of order $\epsilon^2$. Here the readout condition enters: the shares of the orders in the capacity of the readout follow from $\Sigma$ only
if $MQM^{\mathsf T}$ has full rank. For one observable $a$ this requires $a^{\mathsf T}Qa>0$. In step 3,
$\mathcal F_0$ has the same fixed points as $\mathcal T_0$, so $Q$ spans them, and $\ell\circ(1-\mathcal B_0)$ is its
left eigenvector, with value $\ell(\mathcal K_0Q)=s$ on $Q$. Differentiating $(1-\mathcal B)\mathcal F=\mathcal K$ at
$\epsilon=0$ and applying it to $Q$ gives $(1-\mathcal B_0)\mathcal F_0'Q=\mathcal T_0'Q$, whose weight
$\ell(\mathcal T_0'Q)=-\kappa$ is the rate of the slow eigenvalue of $\mathcal T$. The order eigenvalue is
therefore $1-(\kappa/s)\epsilon+O(\epsilon^2)$, and the Laplace transform follows as before. It remains to show that $s$ is positive, which also makes the eigenvalue
simple. Because $\mathcal K_0$ maps positive matrices to positive matrices and $\ell$ is nonnegative on them,
$\ell(\mathcal B_0Y)\le\ell(\mathcal T_0Y)=\ell(Y)$ for every positive $Y$, so $0\le\ell(\mathcal
B_0^j\mathcal K_0Q)\le s\le1$ for every $j$. If $s$ vanished, the series $\mathcal F_0=\sum_j\mathcal B_0^j\mathcal
K_0$ would give $\ell(\mathcal F_0Q)=0$, which contradicts $\mathcal F_0Q=Q$ and $\ell(Q)=1$.

\emph{Numerical check.} We compared the law with exact moments for 24 reset reservoirs and readouts in dimensions two to four. The largest
deviation from the law shrinks in proportion to $p$, from 0.08 at $p=3\%$ to 0.003 at $p=0.1\%$. For five spins the share of the capacity in products of more than $k$ inputs approaches
$e^{-k/(s\,t_{\mathrm{mem}})}$. Its largest deviation, over $k$ and the reservoirs, falls from 0.33 at
$s\,t_{\mathrm{mem}}\approx1$ to 0.07 at $s\,t_{\mathrm{mem}}\approx20$ with the opposite-phase channel, and from
0.43 to 0.14 with independent decay.

\subsection{The scaling limit}
\label{app:gentle}
\emph{Setting.} The amplitude now falls with the leak, $\eta^2=\nu\epsilon$ at fixed $\nu$, and uniformly in the input
\begin{equation}
  \begin{aligned}
    A_u&=O_0\bigl[1+\eta\,a_1(u)+\eta^2a_2(u)-\epsilon Y\bigr]+o(\epsilon),\\
    b_u&=\epsilon\,O_0c_0+o(\epsilon),
  \end{aligned}
  \label{eq:gentleform}
\end{equation}
with a rotation $O_0$, $\mathbb Ea_1=0$, and a step that is a rotation without the leak, so that
$a_1^{\mathsf T}=-a_1$ and $a_2+a_2^{\mathsf T}=a_1^2$. Weak independent decay has this form with $\epsilon=\gamma$,
and the effective Liouvillian of Eq.~\eqref{eq:reduced} has it on the dark coordinates with $\epsilon=1/\gamma$:
$O_0$ is the step of the Hamiltonian at the mean input, $a_1$ and $a_2$ are its response to the input field, and $Y$
and $c_0$ the leak. Let $\Pi_0$ project onto the vectors that $O_0$ leaves fixed and $\Pi$ onto the symmetric matrices
that commute with $O_0$, the slow covariances. For $q$ energy levels with distinct transition frequencies these are
the covariances of the $q-1$ populations and one variance for each coherence, $q(q-1)$ numbers. On them, with
$\bar a_2=\mathbb Ea_2$,
\begin{equation}
  \begin{aligned}
    \mathsf LX&=\Pi(YX+XY^{\mathsf T}),\\
    \mathsf SX&=-\Pi(\bar a_2X+X\bar a_2^{\mathsf T}),\\
    \mathsf KX&=\Pi\,\mathbb E[a_1Xa_1^{\mathsf T}]
  \end{aligned}
  \label{eq:slowmaps}
\end{equation}
are the leak, what the inputs take from a stored product on average, and what they pass on to the next order, with
$\Tr\mathsf SX=\Tr\mathsf KX$.

\begin{keybox}
\textbf{Scaling limit, precise form.} Let every slow covariance decay: the eigenvalues of
$\mathsf L+\nu(\mathsf S-\mathsf K)$ have positive real parts. As $\epsilon\to0$ at fixed $\nu$:
\begin{enumerate}[label=(\alph*),leftmargin=*,itemsep=1pt,topsep=3pt]
\item The stationary mean converges to the vector $m_0$ fixed by $O_0$ with
  $\Pi_0(Y-\nu\bar a_2)m_0=\Pi_0c_0$.
\item With $V_1=\mathbb E[a_1m_0m_0^{\mathsf T}a_1^{\mathsf T}]$ and
  $\mathsf F_\nu=\nu(\mathsf L+\nu\mathsf S)^{-1}\mathsf K$,
  \begin{equation*}
    \begin{aligned}
      \Sigma&\to\Sigma_\nu=\nu\,[\mathsf L+\nu(\mathsf S-\mathsf K)]^{-1}\Pi V_1,\\
      G_k&\to\mathsf F_\nu^{\,k-1}\,\nu\,(\mathsf L+\nu\mathsf S)^{-1}\Pi V_1 .
    \end{aligned}
  \end{equation*}
\item $\mathsf F_\nu$ has a spectral radius $\varrho<1$. For a readout $y=Mr$ with $M\Sigma_\nu M^{\mathsf T}$ of full
  rank, the distribution of $D$ converges to that of (b), with a finite mean and the share of $k=1$ for single inputs.
\item As $\nu\to0$, $\Sigma_\nu$ and $\varrho$ vanish in proportion to $\nu$. As $\nu\to\infty$, if the identity is
  the only slow covariance that the inputs leave unchanged, $\Sigma_\nu$ vanishes as $1/\nu$ and
  $\varrho/(1-\varrho)=\nu\,\mathbb E\|a_1\|_F^2/(2\Tr Y)+O(1)$, the $s\,t_{\mathrm{mem}}$ of the reset model.
\end{enumerate}
\end{keybox}

\emph{Proof.} Over a lifetime $O_0$ turns many times, so only what turns with it accumulates, and on that the leak
and the inputs act at the same order $\epsilon$. This is the weak-coupling limit of open quantum
systems~\cite{Davies1974}, taken for the covariance over input histories. It rests on a lemma: let $R$ be orthogonal,
$P$ the projector onto the vectors it leaves fixed and $M_\epsilon=1-R-\epsilon RN+o(\epsilon)$. If $\Gamma=-PNP$ is
invertible on the range of $P$, then $\epsilon M_\epsilon^{-1}\to\Gamma^{-1}P$. Indeed $PR=RP=P$, so in the blocks of
$P$ and $1-P$ the first diagonal block of $M_\epsilon$ is $\epsilon\Gamma+o(\epsilon)$, the off-diagonal blocks are of
order $\epsilon$, and the second tends to $1-R$, which is invertible there.

With $R=O_0$ the lemma gives (a), since $1-\bar A=1-O_0-\epsilon O_0(\nu\bar a_2-Y)+o(\epsilon)$. On symmetric
matrices $\mathcal R_0X=O_0XO_0^{\mathsf T}$ is orthogonal, and $\Pi$ projects onto its fixed points. With primes for
the maps of Eq.~\eqref{eq:slowmaps} without $\Pi$, Eq.~\eqref{eq:recursion} has
$\mathcal B=\mathcal R_0[1-\epsilon(\mathsf L'+\nu\mathsf S')]$, $\mathcal K=\epsilon\nu\,\mathcal R_0\mathsf K'$
and, because one input writes $\eta\,O_0a_1m_0$, $V=\epsilon\nu\,\mathcal R_0V_1$, each up to $o(\epsilon)$. The
lemma gives $\epsilon(1-\mathcal B)^{-1}\to(\mathsf L+\nu\mathsf S)^{-1}\Pi$ and
$\epsilon(1-\mathcal T)^{-1}\to[\mathsf L+\nu(\mathsf S-\mathsf K)]^{-1}\Pi$, hence (b), and the order transfer
$(1-\mathcal B)^{-1}\mathcal K$ converges to a map whose nonzero eigenvalues are those of $\mathsf F_\nu$. The orders
then fall geometrically at a common rate for small $\epsilon$, which gives (c) once $\varrho<1$.

The semigroups generated by $-(\mathsf L+\nu\mathsf S)$ and $-[\mathsf L+\nu(\mathsf S-\mathsf K)]$ keep positive
matrices positive, as projections with the positive map $\Pi$ of $X\mapsto e^{tW}Xe^{tW^{\mathsf T}}$,
$W=\nu\bar a_2-Y$, without and with the positive map $\nu\mathsf K'$ added to its generator, and the first is bounded
by the second. The second decays by assumption, so both generators have positive inverses, and
$(1-\mathsf F_\nu)^{-1}=1+\nu[\mathsf L+\nu(\mathsf S-\mathsf K)]^{-1}\mathsf K$ is positive, which gives
$\varrho<1$ as in Appendix~\ref{app:proof}. On covariances of fixed vectors $\mathsf L+\nu\mathsf S$ acts as
$X\mapsto W_0X+XW_0^{\mathsf T}$ with $W_0=\Pi_0(Y-\nu\bar a_2)\Pi_0$, so $W_0$ is invertible and $m_0$ exists.

For (d), (b) is linear in $\nu$ to first order. For large $\nu$, $\mathsf F_\nu\to\mathsf S^{-1}\mathsf K$, which
leaves the identity unchanged because $\mathsf S1=\mathsf K1$, with the left eigenvector $X\mapsto\Tr\mathsf SX$. To
first order in $1/\nu$ its eigenvalue moves to $1-2\Tr Y/(\nu\,\mathbb E\|a_1\|_F^2)$. The mean $m_0$ falls as
$1/\nu$, $V_1$ as $1/\nu^2$, and $\Sigma_\nu$, which piles up along the identity as in Eq.~\eqref{eq:slowcov}, as
$1/\nu$. In the reset model $Y=1$ and $s=\eta^2\,\mathbb E\|a_1\|_F^2/d$ to leading order.

\section{One qubit in closed form}
\label{app:qubit}
A qubit relaxes towards $+X$ at a rate $\epsilon$, through the jump $\sqrt\epsilon\,|{+x}\rangle\langle{-x}|$, is
turned about $Z$ by $u\theta$ per step by the inputs $u=\pm1$, and is read along $Y$. The averaged step keeps the axes
$X$ and $Y$ apart and the input-dependent part swaps them, so $Y$ carries only odd orders: order $2j+1$ has the share
$(1-\chi)\chi^j$ of its capacity, where $\chi$ is the factor by which the recursion~\eqref{eq:recursion} returns
covariance from $Y$ through $X$ to $Y$, and $\mathbb ED=(1+\chi)/(1-\chi)$. At a fixed angle with
$\sin\theta\neq0$, $1-\chi=3\epsilon/\sin^2\theta+O(\epsilon^2)$ as $\epsilon\to0$: $\epsilon D$ becomes
exponentially distributed with rate $3/(2\sin^2\theta)$, and the share of single inputs vanishes. If instead
$\theta^2/\epsilon\to\nu$, Appendix~\ref{app:gentle} applies with $O_0=1$ and gives
\begin{equation}
  \chi\to\frac{\nu^2}{(\nu+2)(\nu+1)},\quad
  \mathrm{Var}\,Y\to\frac{4\nu}{(\nu+2)(3\nu+2)} .
  \label{eq:qubitlimit}
\end{equation}
Single inputs keep the share $1-\chi$. The variance of $Y$ is largest at $\nu=2/\sqrt3$, where single inputs carry
80\% of it, and vanishes as $4/(3\nu)$ for $\nu\gg1$ and as $\nu$ for $\nu\ll1$. Direct propagation of the density
matrix reproduces these expressions.

\section{Three spins at strong loss}
\label{app:spins}
\begin{keybox}
\textbf{What is certified.} Take three spins with a collective channel of opposite phases, Eq.~\eqref{eq:H3}, four
Hamiltonians H1--H4 with couplings drawn from fixed seeds, and every Hamiltonian whose three couplings lie in a small
box around theirs. As $\gamma$ grows:
\begin{itemize}[leftmargin=*,itemsep=2pt,topsep=4pt]
\item At a fixed input amplitude, the order of the memory of $X_1$, divided by $\gamma$, becomes exponentially
  distributed, and the variance of $X_1$ falls as $1/\gamma$ (boxes of half-width 1/32).
\item With the scaled amplitude $\eta=8/\sqrt\gamma$, that is $\nu=64$, single past inputs
  (Appendix~\ref{sec:order}) keep a finite share of the capacity of $X_1$,
  and the range of delays that holds 90\% of their capacity grows in proportion to $\gamma$ (boxes of half-width 1/512).
\end{itemize}
\end{keybox}

\paragraph{Model and check.}
The three spins evolve under
\begin{equation}
  H_\eta(u)=\sum_{i<j}J_{ij}X_iX_j+\tfrac12\sum_iZ_i+\tfrac{3+\eta u}{4}\sum_iX_i
  \label{eq:H3}
\end{equation}
with one jump $\sqrt\gamma\,(\sigma_1^--\sigma_2^--\sigma_3^-)$, steps of $\Delta t=1/2$ and signs $u=\pm1$, at
$\eta=1$ or $\eta=8/\sqrt\gamma$. The jump leaves three dark states, with six slow covariances. At strong loss one
step of the full master equation is a rotation of them up to corrections of order $1/\gamma$, so the mixing law
without a reset (Appendix~\ref{app:physical}) gives the first statement and the scaling limit
(Appendix~\ref{app:gentle}) the second, here with explicit error bounds for the full master equation. Both end in
finitely many strict inequalities, which the archive \texttt{spin\_certificates.zip} in the repository checks in
interval arithmetic at 256 bits (ball arithmetic of Arb, python-flint 0.9.0), and the note
\texttt{three\_spin\_limits.pdf} writes out the derivation. An independent program builds the full finite-rate maps
and agrees with the certified values to better than one part in a million. The certified rate of the exponential law
lies between 61 and 133, and with the scaled amplitude at $\gamma=2^{16}$ single past inputs keep between 8\% and 81\% of the
capacity of $X_1$.

\paragraph{Why coherence carries the memory.}
In the limit the mean $m_0$ consists of populations of the energy levels of the Hamiltonian restricted to the dark
states, and a new input turns a population into a coherence between two levels: $V_1$ lies in the coherences.
Coherence shapes the performance of spin reservoirs~\cite{Palacios2024}. Dephasing in this energy basis after every
step removes the coherences, and the stored signal falls to order $\epsilon^2$, that is, as $1/\gamma^2$. Numerically the variance of $X_1$ stays between $10^{-4}$ and
$2\times10^{-3}$ from $\gamma=2^8$ to $2^{18}$ when the coherences are kept and falls as $\gamma^{-2}$, by six orders
of magnitude, when they are erased.

\section{Numerical details and supporting checks}
\label{app:numerics}
\subsection{Memory lifetime and dark states (Sec.~\ref{sec:clock})}
\label{app:num-lifetime}
Figures smooth recall curves for display only, by averaging each delay with its neighbours within a fifth of the
delay, and no sum uses the smoothed values. The maps of Fig.~\ref{fig:clock}a,b are computed at $\gamma=0.1$, 0.25
and $2^{-1}$ to $2^6$ and interpolated along the delay. Figure~\ref{fig:clock}c stops the equal-phase channel and the two-group channel at
$\gamma=64$, where their recall at a fixed amplitude starts to reach past the 300 delays we sum. Summed over all
delays with exact moments (Appendix~\ref{sec:order}), their capacity keeps rising up to $\gamma=256$.

Both tests of the dark states act on the 100 dark coordinates. For the dark-state model we write
$\mathcal P=RR^\flat$ with $R^\flat R=1$, where $R$ embeds the dark coordinates $\alpha$, evolve $\alpha$ under
Eq.~\eqref{eq:reduced} and record the same 45 features of the full state
$(R-\mathcal S\mathcal H(u)R/\gamma)\,\alpha$. The readout is trained on them exactly as for the full reservoir, and
the model has no fitted parameter (Fig.~\ref{fig:darkmodel}a). For the erasure test we replace the dark coordinates
after every step by their time average, $r\mapsto r-RR^\flat(r-\bar r)$. The control erases the dark subspace rotated
by a random product of single-spin unitaries. Recall is measured at delays 1 to 100 with 800 training, 200 validation
and 800 test inputs, at $\gamma=16$ (Fig.~\ref{fig:darkmodel}b).

The decay out of the dark states falls exactly as $1/\gamma$: for the equal-phase channel $\gamma\bar\kappa$
approaches $\kappa_{\mathrm c}=2.945$ from $\gamma=512$ on, so $t_{\mathrm f}$ grows in proportion to $\gamma$, while it
levels off near 56 steps with the opposite-phase channel and near 110 with pairwise decay. The second-order
calculation of Eq.~\eqref{eq:reduced} gives the same $\kappa_{\mathrm c}$ at five spins and 1.4, 2.1 and 3.4 at three, four and six spins: at every size we computed, the decay falls as $1/\gamma$, with a prefactor that
depends on the size.

\begin{figure}[!t]
\centering
\includegraphics[width=\columnwidth]{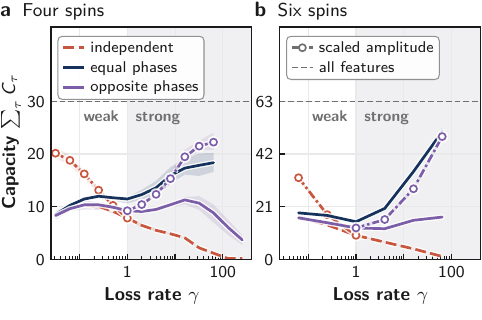}
\caption{\textbf{Other reservoir sizes.} Capacity against the loss rate as in Fig.~\ref{fig:clock}c, for four
spins (eight reservoirs) and six spins (four reservoirs), with the same shading. A line stops where its recall reaches past the 300 delays summed.}
\label{fig:sizes}
\end{figure}

Figure~\ref{fig:sizes} repeats the curves of Fig.~\ref{fig:clock}c at four and six spins, with 30 and 63 features,
on eight and four reservoirs. The opposite-phase channel reverses the sign of the last $\lfloor N/2\rfloor$ spins.
With the equal-phase channel the capacity grows with the loss at both sizes, from 11.5 to 18.4 units between
$\gamma=1$ and 64 at four spins and from 15 to 51 at six, and a scaled amplitude raises the best capacity of the other
two channels. At three spins, which have only three dark states, the capacity of the equal-phase channel does not
grow.

\subsection{Recall time and phases (Sec.~\ref{sec:mixing})}
\label{app:num-mixing}
The recall time $t_{\mathrm f}$ shows the balance of decay and mixing. Its rate adds the decay of the
stored signal and a mixing rate that grows as $\eta^2$ (Appendix~\ref{sec:methods-slow}). The decay grows as $\gamma$
with independent decay and falls as $1/\gamma$ with the opposite-phase channel, so $\eta^2t_{\mathrm f}$ depends on
loss and input only through $\gamma/\eta^2$ and through $\gamma\eta^2$, the $1/\nu$ and $\nu$ of
Appendix~\ref{app:gentle}. Both collapses hold for the four
amplitudes $\eta=1$, $\tfrac12$, $\tfrac14$ and $\tfrac18$ (Fig.~\ref{fig:checks}a), for $\gamma$ from $2^{-8}$ to
$2^6$ with independent decay and from $2^{-2}$ to $2^{12}$ with the opposite-phase channel.

\begin{figure}[!t]
\centering
\includegraphics[width=\columnwidth]{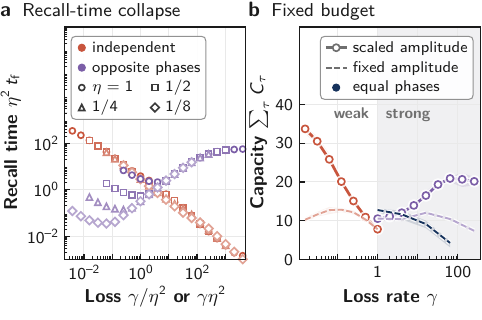}
\caption{\textbf{Two supporting results.}
\textbf{a}~The recall time $t_{\mathrm f}$, rescaled as the sum of decay and mixing predicts: the four amplitudes fall
on one curve per channel.
\textbf{b}~Capacity recovered at a budget of $10^9$ reservoir steps per output step, with every repeated step
counted. Shading as in Fig.~\ref{fig:clock}c.}
\label{fig:checks}
\end{figure}

The change per input of Fig.~\ref{fig:mixing}b is $2\|R^\flat[E(\tfrac34)-E(\tfrac14)]r_*\|$ at the stationary state
$r_*$ of the equal-phase channel, for the input patterns $(1,\dots,1)$ and $\sqrt2\cos(2\pi i/5)$ multiplying $X_i$.
The mode-matched pattern falls as $1/\gamma$, to 0.003 at $\gamma=128$, while the other falls from 0.40 to 0.18. The rule
is sharp: a phase error $\delta$ couples dark states directly at order $\delta$, so beyond $\gamma\sim1/\delta$ the
equal-phase channel behaves like the opposite-phase channel.

\subsection{Orders and capacities (Sec.~\ref{sec:mixing})}
\label{app:num-capacity}
The distribution of the number of inputs mixed, Eq.~\eqref{eq:recursion}, uses six-point Gauss--Legendre quadrature
for the expectations over the input and solves the equations $X-\bar AX\bar A^{\mathsf T}=C$ in the eigenbasis of
$\bar A$. The shares of the 45 features are summed until the rest falls below $10^{-8}$, or until successive orders
shrink at a settled ratio, whose geometric tail we add. The recall follows from the same moments,
$C_\tau=c_\tau^{\mathsf T}\Sigma_y^{-1}c_\tau/\mathrm{Var}\,u$ with $c_\tau$ the covariance of the features with
$u_{t-\tau}$ and $\Sigma_y$ their covariance, and reproduces the simulated capacities to a median of 2\%. Three steps of
inverse iteration on $\mathcal T$, from the input's covariance $V$ and from the identity, give $t_{\mathrm{mem}}$, $Q$
and $\ell$. Loss rates run in powers of four over the ranges of Fig.~\ref{fig:mixing}c, and with the scaled amplitude to
16384 and $1/4096$, on four reservoirs, two at the longest memory lifetimes. At a fixed amplitude the mean order of the
opposite-phase channel rises from 1.8 at $\gamma=1$ to 23 at $\gamma=4096$, and with independent decay it stays within
2.2 of $s\,t_{\mathrm{mem}}$ up to $s\,t_{\mathrm{mem}}=18$.

Figure~\ref{fig:mixing}a shows the linear memory capacity over all delays, from the same moments. At a fixed
amplitude it peaks at 16.0 units with the opposite-phase channel, at $\gamma=64$ and $s\,t_{\mathrm{mem}}=0.32$, and
falls to 2.2 units at a lifetime of 750 steps, where single inputs carry 0.07 of the total capacity. With independent
decay it peaks at 14.9 units, at $\gamma=1/16$ and $s\,t_{\mathrm{mem}}=0.20$, and falls to 3.7 units at 430 steps,
where single inputs carry 0.10. With the scaled amplitude single inputs carry 0.85 of the total capacity at a
lifetime of 20 steps and 0.88 to 0.92 beyond 80 steps, where the capacity stays between 39 and 40 units and
$s\,t_{\mathrm{mem}}$ between 0.007 and 0.014.

\paragraph{The scaling limit.}
We build the limit of Appendix~\ref{app:gentle} from the Liouvillian of each reservoir: $O_0$, $a_1$ and $a_2$ are the
orders zero to two of the step in the input field, and $Y$ and $c_0$ its first order in the leak. Independent decay
has 992 slow covariances and the opposite-phase channel 90, for the ten dark states under Eq.~\eqref{eq:reduced}, and
at $\nu=1$ they decay by at least $0.5\,\epsilon$ and $1.4\,\epsilon$ per step. With the opposite-phase channel 33
directions of the 45 features respond at leading order, 11 through the admixture of decaying operators of order
$1/\gamma$, which the readout of the limit includes and which raises the share of single inputs by up to 0.02, and
one at order $1/\gamma^2$, which the recursion no longer resolves from $\gamma=1024$ on.

For the effective Liouvillian of Eq.~\eqref{eq:reduced} the recursion converges onto the limit: at $\nu=1$ the share of single inputs agrees to
four digits from $\gamma=2^{14}$ on. For the full master equation at $\gamma=1/4096$ and 4096, on two reservoirs and
for $\nu=1/16$ to 64, it differs from the limit by at most 0.002 with independent decay and 0.009 with the
opposite-phase channel, the mean of $D$ by 0.003 and 0.016, and the variance that single inputs leave in the 45
features by 0.4\%. At $\nu=1$ the limit gives single inputs 0.86 to 0.93 of the total capacity over four reservoirs
with independent decay and 0.90 to 0.94 with the opposite-phase channel. The variance that single inputs leave is
largest at $\nu=22$ to 46 with independent decay and at $\nu=68$ to 215 with the opposite-phase channel. There single
inputs carry 0.38 to 0.55 of the total capacity, and $\varrho/(1-\varrho)$ lies between 0.8 and 4.8.

Table~\ref{tab:memory} collects the capacities of simulated runs on eight reservoirs. Their readout is trained and
tested on a finite number of inputs, and their sum stops at 300 delays. On the four reservoirs that both use, with the
scaled amplitude at $\gamma=1/64$ and 256, the simulated runs give 35.0 and 37.1 units, and the exact moments give 36.8
and 39.4 units within the first 300 delays and 36.9 and 39.9 over all delays. All capacities of a readout add up to its rank, 45 here~\cite{Dambre2012}, so raising the linear
memory capacity from 15 to 40 units leaves at most 5 instead of 30 for products of inputs.

\begin{table}[!t]
\caption{\label{tab:memory}Linear memory capacity $\sum_\tau C_\tau$ of the five-spin reservoir in simulated runs, in
units (maximum 45) and summed over the first 300 delays: with exact expectation values, at a budget of $10^9$
reservoir steps per output step, and with $n$ shots per setting and step with the rewinding protocol
(Sec.~\ref{sec:price}). Mean and standard deviation over eight reservoirs. Exact moments over all delays, as quoted in
Sec.~\ref{sec:mixing}, are higher by up to 4 units (Appendix~\ref{app:num-capacity}).}
\centering\footnotesize\setlength{\tabcolsep}{3pt}
\begin{tabular}{@{}lccc@{}}
\toprule
 & independent & opposite & equal\\
 & decay & phases & phases\\\midrule
fixed amplitude$^{a}$ & $14.5\pm1.0$ & $14.7\pm1.1$ & $30.6\pm1.4$\\
scaled amplitude$^{b}$ & $34.9\pm2.0$ & $35.7\pm3.0$ & --$^{c}$\\
budget $10^9$, fixed$^{d}$ & $12.8\pm0.7$ & $12.0\pm0.7$ & $12.8\pm0.7$\\
budget $10^9$, scaled$^{e}$ & $33.7\pm1.1$ & $20.9\pm1.1$ & --$^{c}$\\
$\gamma=64$, exact$^{f}$ & $1.2\pm0.0$ & $14.7\pm1.1$ & $30.6\pm1.4$\\
$\gamma=64$, $n=10^4$$^{f}$ & $0.2\pm0.0$ & $8.2\pm0.8$ & $0.2\pm0.0$\\
$\gamma=64$, $n=10^6$$^{f}$ & $0.2\pm0.0$ & $9.9\pm0.7$ & $2.3\pm1.0$\\
\bottomrule
\end{tabular}

\smallskip
\raggedright\footnotesize
$^{a}$At the loss rates with the largest sum over 300 delays, $\gamma=0.1$, 64 and 64. Over all delays the capacity of
the equal-phase channel keeps rising up to $\gamma=256$ (Appendix~\ref{app:num-lifetime}). $^{b}$At $\gamma=1/64$ and 256. $^{c}$The uniform input is mode-matched
to this channel. $^{d}$At $\gamma=1/8$, 16 and 1, the best loss rates on the grid. $^{e}$At $\gamma=1/64$ and 64.
$^{f}$At a fixed amplitude.
\end{table}

\subsection{Best input amplitude (Secs.~\ref{sec:mixing} and~\ref{sec:price})}
\label{app:num-amplitude}
We mapped the linear memory capacity over loss rates and amplitudes: the first four reservoirs, independent decay at
$\gamma=2^{-8}$ to $2^{-2}$ and at 1, the opposite-phase channel at $\gamma=1$, 4 and $2^3$ to $2^8$, and fifteen
amplitudes from $\eta=2^{-7}$ to 1, with exact features or $10^6$ to $10^8$ shots per setting and step from
the state with the full history. With exact expectation values the capacity rises as the input weakens, up to 43 of the 45
units. With $10^6$ to $10^8$ shots an optimum exists. Its amplitude falls with the relaxation time, computed as the recall
time at $\eta=2^{-8}$, with fitted exponents between $-0.40$ and $-0.46$ in all six fits over the
loss rates whose relaxation time is at least four steps. Without the loss rates whose recall reaches past the 300 delays
summed, the exponents are $-0.48$ to $-0.56$ with independent decay and $-0.38$ to $-0.39$ with the opposite-phase
channel. The optimum is flat. Over the loss rates of the map, one amplitude recovers on average 91\% to 97\% of the
capacity at the optimum tracked at each loss rate, so within such a range the amplitude needs no fine tuning.

\subsection{Measurement (Sec.~\ref{sec:price})}
\label{app:num-measurement}
With the opposite-phase channel, the scaled amplitude lowers the variance of the features from $8\times10^{-5}$ to $5\times10^{-6}$,
and restarting each run from $|00000\rangle$, as the rewinding protocol does, raises it to $4\times10^{-5}$, against
$1.2\times10^{-4}$ with independent decay and the scaled amplitude at $\gamma=1/32$.

The rewinding window $\tau_{90}$ grows with the memory lifetime: with the scaled amplitude of the
opposite-phase channel from about 30 steps at $\gamma=16$ to 130 at $\gamma=256$. The thresholds of Sec.~\ref{sec:price}
are the crossings of the capacities at a fixed and with the scaled amplitude, from runs with the rewinding protocol at $10^3$ to
$10^{8.5}$ shots per setting and step. With the rewinding protocol the equal-phase channel at $\gamma=64$ stays at the
floor up to $10^4$ shots and recovers 2.3 units at $10^6$ and 15.5 at $3\times10^8$. There it passes the
opposite-phase channel at a fixed amplitude, 11.0 units, and stays below it with the scaled amplitude, 26.3 units.
At $10^3$ and $10^4$ shots the opposite-phase channel at $\gamma=64$ recovers 5.4 and 8.2 units at a fixed amplitude,
and independent decay at $\gamma=1/32$ recovers 8.2 and 11.0. The scaled amplitude overtakes the fixed one at
$2\times10^4$ shots with independent decay and at $10^5$ with the opposite-phase channel, and independent decay then
reaches 31.5 units at $3\times10^8$.

Figure~\ref{fig:checks}b shows the capacity at a budget of $10^9$ reservoir steps per output step, which is
$3\times10^6$ to $3\times10^7$ shots per setting and step, and Table~\ref{tab:memory} lists its best values.
With the full history, eight times more
training inputs change the sampled capacity by less than 1.3 units.

\begin{table}[!t]
\caption{\label{tab:matched}The eight decay channels at a relaxation time of 32 input steps and $\eta=0.35$:
loss rate, linear response of the stationary state, and linear memory capacity with exact expectation values and with
$10^6$ shots per measurement
setting and step. Mean and standard deviation over eight reservoirs.}
\centering\footnotesize\setlength{\tabcolsep}{3.5pt}
\begin{tabular}{@{}lcccc@{}}
\toprule
 & loss & & \multicolumn{2}{c}{capacity}\\
channel & rate $\gamma$ & response & exact & at $10^6$\\\midrule
independent & 0.092 & 1.3 & $23.7\pm1.1$ & $17.4\pm2.3$\\
exchange & 0.064 & 0.30 & $25.1\pm0.8$ & $10.1\pm0.9$\\
gain/loss & 0.079 & 0.46 & $24.4\pm0.8$ & $12.5\pm0.9$\\
pairwise & 168 & 1.5 & $18.2\pm2.0$ & $12.4\pm1.2$\\
opposite phases & 101 & 0.68 & $25.3\pm2.6$ & $15.0\pm1.5$\\
two groups & 73 & 0.011 & $26.6\pm2.9$ & $0.3\pm0.1$\\
equal phases & 47 & 0.018 & $35.3\pm0.8$ & $0.3\pm0.1$\\
dephasing & 143 & 0 & 0.0 & 0.2$^{a}$\\
\bottomrule
\end{tabular}

\smallskip
\raggedright\footnotesize
$^{a}$A readout without any recall reads about 0.2 units, the floor of the sum over 300 delays.
\end{table}

For the matched relaxation time of Fig.~\ref{fig:sampling}b,c and Table~\ref{tab:matched} we find, for every channel
and reservoir, the loss rate at which $t_{\mathrm{rel}}=32$ steps, at weak loss when the kernel of the dissipator is
the stationary state alone and at strong loss otherwise, and run the reservoir there with $\eta=\sqrt{4/32}$. The
linear response of the stationary state is the Euclidean norm of the Pauli
coefficients $\Tr(PW)$, over all Pauli strings $P$, of the change $W$: the commutator $-i[H_{\mathrm{in}},\rho_*]$ with
$H_{\mathrm{in}}=h\sum_iX_i$ and $\rho_*$ the stationary state, projected spectrally onto the slow modes. These are the
modes of the kernel, or without one those decaying at most four times faster than the slowest. The loss
rate is matched for each reservoir, so Table~\ref{tab:matched} gives its mean. Over the 56 reservoirs of the seven channels that respond, the rank correlation of the linear response is 0.96 with
the recovered fraction of the capacity and 0.83 with the recovered capacity itself. The seven channel medians have a
rank correlation of 0.93 with the recovered fraction. Within single channels, over eight reservoirs each, it is 0.81
to 0.95 for six channels and 0.29 for pairwise decay, and with the variance of the features it is at least 0.83 for
all seven.

\subsection{Added independent decay and dephasing}
\label{app:num-noise}
Adding independent decay $\sqrt{\gamma_{\mathrm{ind}}}\,\sigma_i^-$ to the equal-phase channel gives the dark modes
the mean decay rate $\bar\kappa_0=\kappa_{\mathrm c}/\gamma+\kappa_{\mathrm l}\gamma_{\mathrm{ind}}$, where to first
order $\kappa_{\mathrm l}$ is the trace of the added dissipator on the 99 slow dark modes divided by 99,
$140/99\approx1.414$. Dephasing $\sqrt{\gamma_\phi}\,Z_i$ adds $\kappa_\phi\gamma_\phi$ with
$\kappa_\phi=403.2/99\approx4.073$, and the full spectra reproduce both coefficients to within 0.2\%. A relaxation
time of at least $T$ steps, $\bar\kappa_0\le1/(\Delta t\,T)$, then needs
$\gamma/\gamma_{\mathrm{ind}}\ge4\kappa_{\mathrm c}\kappa_{\mathrm l}\,\Delta t^2T^2\approx4.2\,T^2$, or $12\,T^2$ for
dephasing.

%% file: dissipation_qrc.bbl
\begin{thebibliography}{10}

\bibitem{Maass2002}
Wolfgang Maass, Thomas Natschl{\"a}ger, and Henry Markram.
\newblock ``Real-time computing without stable states: A new framework for
  neural computation based on perturbations''.
\newblock \href{https://dx.doi.org/10.1162/089976602760407955}{Neural
  Computation {\bf 14}, 2531--2560}~(2002).

\bibitem{JaegerESP2001}
Herbert Jaeger.
\newblock ``The ``echo state'' approach to analysing and training recurrent
  neural networks''.
\newblock \href{https://dx.doi.org/10.24406/publica-fhg-291111}{Technical
  Report 148}.
\newblock GMD Forschungszentrum Informationstechnik~(2001).

\bibitem{Fujii2017}
Keisuke Fujii and Kohei Nakajima.
\newblock ``Harnessing disordered-ensemble quantum dynamics for machine
  learning''.
\newblock \href{https://dx.doi.org/10.1103/PhysRevApplied.8.024030}{Physical
  Review Applied {\bf 8}, 024030}~(2017).

\bibitem{ChenNurdin2019}
Jiayin Chen and Hendra~I. Nurdin.
\newblock ``Learning nonlinear input--output maps with dissipative quantum
  systems''.
\newblock \href{https://dx.doi.org/10.1007/s11128-019-2311-9}{Quantum
  Information Processing {\bf 18}, 198}~(2019).

\bibitem{Chen2020}
Jiayin Chen, Hendra~I. Nurdin, and Naoki Yamamoto.
\newblock ``Temporal information processing on noisy quantum computers''.
\newblock \href{https://dx.doi.org/10.1103/PhysRevApplied.14.024065}{Physical
  Review Applied {\bf 14}, 024065}~(2020).

\bibitem{Sannia2024}
Antonio Sannia, Rodrigo Mart{\'i}nez-Pe{\~n}a, Miguel~C. Soriano, Gian~Luca
  Giorgi, and Roberta Zambrini.
\newblock ``Dissipation as a resource for quantum reservoir computing''.
\newblock \href{https://dx.doi.org/10.22331/q-2024-03-20-1291}{Quantum {\bf 8},
  1291}~(2024).

\bibitem{Ricci2026}
Emanuele Ricci, Francesco Monzani, Luca Nigro, and Enrico Prati.
\newblock ``Quantum reservoir computing induced by controllable damping''.
\newblock \href{https://dx.doi.org/10.1038/s41534-026-01229-8}{npj Quantum
  Information {\bf 12}, 107}~(2026).

\bibitem{Domingo2023}
Laia Domingo, Gabriel Carlo, and Florentino Borondo.
\newblock ``Taking advantage of noise in quantum reservoir computing''.
\newblock \href{https://dx.doi.org/10.1038/s41598-023-35461-5}{Scientific
  Reports {\bf 13}, 8790}~(2023).

\bibitem{Goetting2025}
Niclas G{\"o}tting, Steffen Wilksen, Alexander Steinhoff, Frederik Lohof, and
  Christopher Gies.
\newblock ``Connection between memory performance and optical absorption in
  quantum reservoir computing''.
\newblock \href{https://dx.doi.org/10.1103/vp79-8t1l}{Physical Review Letters
  {\bf 135}, 240403}~(2025).

\bibitem{Mujal2023}
Pere Mujal, Rodrigo Mart{\'i}nez-Pe{\~n}a, Gian~Luca Giorgi, Miguel~C. Soriano,
  and Roberta Zambrini.
\newblock ``Time-series quantum reservoir computing with weak and projective
  measurements''.
\newblock \href{https://dx.doi.org/10.1038/s41534-023-00682-z}{npj Quantum
  Information {\bf 9}, 16}~(2023).

\bibitem{Franceschetto2026}
Giacomo Franceschetto, Marcin P{\l}odzie{\'n}, Maciej Lewenstein, Antonio
  Ac{\'i}n, and Pere Mujal.
\newblock ``Harnessing quantum backaction for time-series processing''.
\newblock \href{https://dx.doi.org/10.1103/j7f9-hfsj}{Physical Review X {\bf
  16}, 021002}~(2026).

\bibitem{KobayashiFeedback2024}
Kaito Kobayashi, Keisuke Fujii, and Naoki Yamamoto.
\newblock ``Feedback-driven quantum reservoir computing for time-series
  analysis''.
\newblock \href{https://dx.doi.org/10.1103/PRXQuantum.5.040325}{PRX Quantum
  {\bf 5}, 040325}~(2024).

\bibitem{Cindrak2024}
Saud {\v C}indrak, Brecht Donvil, Kathy L{\"u}dge, and Lina Jaurigue.
\newblock ``Enhancing the performance of quantum reservoir computing and
  solving the time-complexity problem by artificial memory restriction''.
\newblock \href{https://dx.doi.org/10.1103/PhysRevResearch.6.013051}{Physical
  Review Research {\bf 6}, 013051}~(2024).

\bibitem{Sannia2026Memory}
Antonio Sannia, Ricard~Ravell Rodr{\'i}guez, Gian~Luca Giorgi, and Roberta
  Zambrini.
\newblock ``Non-{M}arkovianity and memory enhancement in quantum reservoir
  computing''.
\newblock \href{https://dx.doi.org/10.1038/s41534-026-01257-4}{npj Quantum
  Information {\bf 12}, 111}~(2026).

\bibitem{Paparelle2026}
Iris Paparelle, Johan Henaff, Jorge Garc{\'i}a-Beni, {\'E}milie Gillet, Daniel
  Montesinos, Gian~Luca Giorgi, Miguel~C. Soriano, Roberta Zambrini, and
  Valentina Parigi.
\newblock ``Experimental memory control in continuous-variable optical quantum
  reservoir computing''.
\newblock \href{https://dx.doi.org/10.1038/s41566-026-01880-9}{Nature Photonics
  {\bf 20}, 413--420}~(2026).

\bibitem{MartinezPenaPRL2021}
Rodrigo Mart{\'i}nez-Pe{\~n}a, Gian~Luca Giorgi, Johannes Nokkala, Miguel~C.
  Soriano, and Roberta Zambrini.
\newblock ``Dynamical phase transitions in quantum reservoir computing''.
\newblock \href{https://dx.doi.org/10.1103/PhysRevLett.127.100502}{Physical
  Review Letters {\bf 127}, 100502}~(2021).

\bibitem{Baumann2026}
Markus Baumann, Itamar Fink, Johannes Wittmann, Claudia Linnhoff-Popien, and
  Jonas Stein.
\newblock ``Where a quantum reservoir works: A transferable operating
  band''~(2026).
\newblock  \href{http://arxiv.org/abs/2606.13284}{arXiv:2606.13284}.

\bibitem{Poyatos1996}
J.~F. Poyatos, J.~Ignacio Cirac, and P.~Zoller.
\newblock ``Quantum reservoir engineering with laser cooled trapped ions''.
\newblock \href{https://dx.doi.org/10.1103/PhysRevLett.77.4728}{Physical Review
  Letters {\bf 77}, 4728--4731}~(1996).

\bibitem{Verstraete2009}
Frank Verstraete, Michael~M. Wolf, and J.~Ignacio Cirac.
\newblock ``Quantum computation and quantum-state engineering driven by
  dissipation''.
\newblock \href{https://dx.doi.org/10.1038/nphys1342}{Nature Physics {\bf 5},
  633--636}~(2009).

\bibitem{Pastawski2011}
Fernando Pastawski, Lucas Clemente, and J.~Ignacio Cirac.
\newblock ``Quantum memories based on engineered dissipation''.
\newblock \href{https://dx.doi.org/10.1103/PhysRevA.83.012304}{Physical Review
  A {\bf 83}, 012304}~(2011).

\bibitem{Harrington2022}
Patrick~M. Harrington, Erich~J. Mueller, and Kater~W. Murch.
\newblock ``Engineered dissipation for quantum information science''.
\newblock \href{https://dx.doi.org/10.1038/s42254-022-00494-8}{Nature Reviews
  Physics {\bf 4}, 660--671}~(2022).

\bibitem{Dicke1954}
Robert~H. Dicke.
\newblock ``Coherence in spontaneous radiation processes''.
\newblock \href{https://dx.doi.org/10.1103/PhysRev.93.99}{Physical Review {\bf
  93}, 99--110}~(1954).

\bibitem{AsenjoGarcia2017}
A.~Asenjo-Garcia, M.~Moreno-Cardoner, A.~Albrecht, H.~J. Kimble, and D.~E.
  Chang.
\newblock ``Exponential improvement in photon storage fidelities using
  subradiance and ``selective radiance'' in atomic arrays''.
\newblock \href{https://dx.doi.org/10.1103/PhysRevX.7.031024}{Physical Review X
  {\bf 7}, 031024}~(2017).

\bibitem{Facchinetti2016}
G.~Facchinetti, S.~D. Jenkins, and J.~Ruostekoski.
\newblock ``Storing light with subradiant correlations in arrays of atoms''.
\newblock \href{https://dx.doi.org/10.1103/PhysRevLett.117.243601}{Physical
  Review Letters {\bf 117}, 243601}~(2016).

\bibitem{FacchiPascazio2002}
P.~Facchi and S.~Pascazio.
\newblock ``Quantum {Zeno} subspaces''.
\newblock \href{https://dx.doi.org/10.1103/PhysRevLett.89.080401}{Physical
  Review Letters {\bf 89}, 080401}~(2002).

\bibitem{Zanardi2014}
Paolo Zanardi and Lorenzo Campos~Venuti.
\newblock ``Coherent quantum dynamics in steady-state manifolds of strongly
  dissipative systems''.
\newblock \href{https://dx.doi.org/10.1103/PhysRevLett.113.240406}{Physical
  Review Letters {\bf 113}, 240406}~(2014).

\bibitem{Burgarth2019}
Daniel Burgarth, Paolo Facchi, Hiromichi Nakazato, Saverio Pascazio, and Kazuya
  Yuasa.
\newblock ``Generalized adiabatic theorem and strong-coupling limits''.
\newblock \href{https://dx.doi.org/10.22331/q-2019-06-12-152}{Quantum {\bf 3},
  152}~(2019).

\bibitem{Reiter2012}
Florentin Reiter and Anders~S. S{\o}rensen.
\newblock ``Effective operator formalism for open quantum systems''.
\newblock \href{https://dx.doi.org/10.1103/PhysRevA.85.032111}{Physical Review
  A {\bf 85}, 032111}~(2012).

\bibitem{Dambre2012}
Joni Dambre, David Verstraeten, Benjamin Schrauwen, and Serge Massar.
\newblock ``Information processing capacity of dynamical systems''.
\newblock \href{https://dx.doi.org/10.1038/srep00514}{Scientific Reports {\bf
  2}, 514}~(2012).

\bibitem{Inubushi2017}
Masanobu Inubushi and Kazuyuki Yoshimura.
\newblock ``Reservoir computing beyond memory-nonlinearity trade-off''.
\newblock \href{https://dx.doi.org/10.1038/s41598-017-10257-6}{Scientific
  Reports {\bf 7}, 10199}~(2017).

\bibitem{Cindrak2026}
Saud {\v C}indrak, Lara Giebeler, Niclas G{\"o}tting, Christopher Gies, and
  Kathy L{\"u}dge.
\newblock ``Memory-nonlinearity trade-off across quantum reservoir computing
  frameworks''~(2026).
\newblock  \href{http://arxiv.org/abs/2603.21371}{arXiv:2603.21371}.

\bibitem{Hu2023}
Fangjun Hu, Gerasimos Angelatos, Saeed~A. Khan, Marti Vives, Esin T{\"u}reci,
  Leon Bello, Graham~E. Rowlands, Guilhem~J. Ribeill, and Hakan~E. T{\"u}reci.
\newblock ``Tackling sampling noise in physical systems for machine learning
  applications: Fundamental limits and eigentasks''.
\newblock \href{https://dx.doi.org/10.1103/PhysRevX.13.041020}{Physical Review
  X {\bf 13}, 041020}~(2023).

\bibitem{White2004}
Olivia~L. White, Daniel~D. Lee, and Haim Sompolinsky.
\newblock ``Short-term memory in orthogonal neural networks''.
\newblock \href{https://dx.doi.org/10.1103/PhysRevLett.92.148102}{Physical
  Review Letters {\bf 92}, 148102}~(2004).

\bibitem{Ganguli2008}
Surya Ganguli, Dongsung Huh, and Haim Sompolinsky.
\newblock ``Memory traces in dynamical systems''.
\newblock \href{https://dx.doi.org/10.1073/pnas.0804451105}{Proceedings of the
  National Academy of Sciences {\bf 105}, 18970--18975}~(2008).

\bibitem{Jaeger2007Leaky}
Herbert Jaeger, Mantas Luko{\v s}evi{\v c}ius, Dan Popovici, and Udo Siewert.
\newblock ``Optimization and applications of echo state networks with
  leaky-integrator neurons''.
\newblock \href{https://dx.doi.org/10.1016/j.neunet.2007.04.016}{Neural
  Networks {\bf 20}, 335--352}~(2007).

\bibitem{Kubota2023}
Tomoyuki Kubota, Yudai Suzuki, Shumpei Kobayashi, Quoc~Hoan Tran, Naoki
  Yamamoto, and Kohei Nakajima.
\newblock ``Temporal information processing induced by quantum noise''.
\newblock \href{https://dx.doi.org/10.1103/PhysRevResearch.5.023057}{Physical
  Review Research {\bf 5}, 023057}~(2023).

\bibitem{Fry2023}
Daniel Fry, Amol Deshmukh, Samuel Yen-Chi Chen, Vladimir Rastunkov, and Vanio
  Markov.
\newblock ``Optimizing quantum noise-induced reservoir computing for nonlinear
  and chaotic time series prediction''.
\newblock \href{https://dx.doi.org/10.1038/s41598-023-45015-4}{Scientific
  Reports {\bf 13}, 19326}~(2023).

\bibitem{Cheamsawat2025}
Krai Cheamsawat and Thiparat Chotibut.
\newblock ``Dissipation alters modes of information encoding in small quantum
  reservoirs near criticality''.
\newblock \href{https://dx.doi.org/10.3390/e27010088}{Entropy {\bf 27},
  88}~(2025).

\bibitem{Akshay2026}
V.~V. Akshay, M.~V. Altaisky, and N.~E. Kaputkina.
\newblock ``Quantum reservoir networks based on decoherence-free
  subspaces''~(2026).
\newblock  \href{http://arxiv.org/abs/2605.27427}{arXiv:2605.27427}.

\bibitem{Gorini1976}
Vittorio Gorini, Andrzej Kossakowski, and E.~C.~George Sudarshan.
\newblock ``Completely positive dynamical semigroups of {N}-level systems''.
\newblock \href{https://dx.doi.org/10.1063/1.522979}{Journal of Mathematical
  Physics {\bf 17}, 821--825}~(1976).

\bibitem{Lindblad1976}
G{\"o}ran Lindblad.
\newblock ``On the generators of quantum dynamical semigroups''.
\newblock \href{https://dx.doi.org/10.1007/BF01608499}{Communications in
  Mathematical Physics {\bf 48}, 119--130}~(1976).

\bibitem{Sannia2025Concentration}
Antonio Sannia, Gian~Luca Giorgi, and Roberta Zambrini.
\newblock ``Exponential concentration and symmetries in quantum reservoir
  computing''~(2025).
\newblock  \href{http://arxiv.org/abs/2505.10062}{arXiv:2505.10062}.

\bibitem{Beige2000}
Almut Beige, Daniel Braun, Ben Tregenna, and Peter~L. Knight.
\newblock ``Quantum computing using dissipation to remain in a decoherence-free
  subspace''.
\newblock \href{https://dx.doi.org/10.1103/PhysRevLett.85.1762}{Physical Review
  Letters {\bf 85}, 1762--1765}~(2000).

\bibitem{Zanner2022}
Maximilian Zanner, Tuure Orell, Christian M.~F. Schneider, Romain Albert,
  Stefan Oleschko, Mathieu~L. Juan, Matti Silveri, and Gerhard Kirchmair.
\newblock ``Coherent control of a multi-qubit dark state in waveguide quantum
  electrodynamics''.
\newblock \href{https://dx.doi.org/10.1038/s41567-022-01527-w}{Nature Physics
  {\bf 18}, 538--543}~(2022).

\bibitem{Mirhosseini2019}
Mohammad Mirhosseini, Eunjong Kim, Xueyue Zhang, Alp Sipahigil, Paul~B.
  Dieterle, Andrew~J. Keller, Ana Asenjo-Garcia, Darrick~E. Chang, and Oskar
  Painter.
\newblock ``Cavity quantum electrodynamics with atom-like mirrors''.
\newblock \href{https://dx.doi.org/10.1038/s41586-019-1196-1}{Nature {\bf 569},
  692--697}~(2019).

\bibitem{MartinezPenaIPC2023}
Rodrigo Mart{\'i}nez-Pe{\~n}a, Johannes Nokkala, Gian~Luca Giorgi, Roberta
  Zambrini, and Miguel~C. Soriano.
\newblock ``Information processing capacity of spin-based quantum reservoir
  computing systems''.
\newblock \href{https://dx.doi.org/10.1007/s12559-020-09772-y}{Cognitive
  Computation {\bf 15}, 1440--1451}~(2023).

\bibitem{GrigoryevaOrtegaSAS2018}
Lyudmila Grigoryeva and Juan-Pablo Ortega.
\newblock ``Universal discrete-time reservoir computers with stochastic inputs
  and linear readouts using non-homogeneous state-affine systems''.
\newblock Journal of Machine Learning Research {\bf 19}, 1--40~(2018).
\newblock  url:~\url{https://jmlr.org/papers/v19/18-020.html}.

\bibitem{Hoeffding1948}
Wassily Hoeffding.
\newblock ``A class of statistics with asymptotically normal distribution''.
\newblock \href{https://dx.doi.org/10.1214/aoms/1177730196}{The Annals of
  Mathematical Statistics {\bf 19}, 293--325}~(1948).

\bibitem{MartinezPenaOrtega2023}
Rodrigo Mart{\'i}nez-Pe{\~n}a and Juan-Pablo Ortega.
\newblock ``Quantum reservoir computing in finite dimensions''.
\newblock \href{https://dx.doi.org/10.1103/PhysRevE.107.035306}{Physical Review
  E {\bf 107}, 035306}~(2023).

\bibitem{MartinezPenaOrtega2025}
Rodrigo Mart{\'i}nez-Pe{\~n}a and Juan-Pablo Ortega.
\newblock ``Input-dependence in quantum reservoir computing''.
\newblock \href{https://dx.doi.org/10.1103/3775-4hfd}{Physical Review E {\bf
  111}, 065306}~(2025).

\bibitem{Davies1974}
E.~B. Davies.
\newblock ``Markovian master equations''.
\newblock \href{https://dx.doi.org/10.1007/BF01608389}{Communications in
  Mathematical Physics {\bf 39}, 91--110}~(1974).

\bibitem{Palacios2024}
Ana Palacios, Rodrigo Mart{\'i}nez-Pe{\~n}a, Miguel~C. Soriano, Gian~Luca
  Giorgi, and Roberta Zambrini.
\newblock ``Role of coherence in many-body quantum reservoir computing''.
\newblock \href{https://dx.doi.org/10.1038/s42005-024-01859-4}{Communications
  Physics {\bf 7}, 369}~(2024).

\end{thebibliography}
